\documentclass{llncs}

\usepackage{latexsym,amsfonts}
\usepackage{amsmath}
\usepackage{amssymb}
\usepackage{array}

\usepackage{graphicx,color}
\usepackage{latexsym,amsfonts}
\usepackage{amsmath}
\usepackage{amssymb}
\usepackage{bbm}

\usepackage{url}

\newcounter{lpnumber} 
\newenvironment{linearprogram}[1]
{ \stepcounter{lpnumber}
  \begin{gather} #1 \tag{LP\arabic{lpnumber}} \\[-5ex] \notag
  \end{gather}
  \hspace{1.5cm} subject to \\[-3ex]
  \align }
{ \endalign }
\newcommand{\minimize}[1]{\text{minimize} \ #1}

\newcommand{\opt}{\mathsf{OPT}}
\newcommand{\Comp}{\mathcal{C}}

\newcommand{\Nbr}{\mathsf{Nbr}}
\newcommand{\wt}{\mathsf{cost}}
\newcommand{\vote}{\mathsf{vote}}
\newtheorem{new-claim}{Claim}

\begin{document}
\title{Max-size popular matchings and extensions}
\author{Telikepalli Kavitha}
\institute{Tata Institute of Fundamental Research, Mumbai, India\\
\email{kavitha@tcs.tifr.res.in}}
\maketitle
\pagestyle{plain}

\begin{abstract}
  We consider the max-size popular matching problem in a roommates instance $G = (V,E)$ with
  strict preference lists. A matching $M$ is popular if there is no matching $M'$ in $G$
  such that the vertices that prefer $M'$ to $M$ outnumber those that prefer $M$ to $M'$.
  We show it is  $\mathsf{NP}$-hard to compute a max-size popular matching in $G$.  This is in contrast to the tractability of
  this problem in bipartite graphs where a max-size popular matching can be computed in linear time.
  We define a subclass of max-size popular matchings called {\em strongly dominant} matchings and show a linear time
  algorithm to solve the strongly dominant matching problem in a roommates instance.

  We consider a generalization of the max-size popular matching problem in bipartite graphs: this is the
  max-weight popular matching problem where there is also a weight function $w: E \rightarrow \mathbb{R}$ and
  we seek a popular matching of largest weight.
  We show this is an $\mathsf{NP}$-hard problem and this is so even when $w(e) \in \{1,2\}$ for every $e \in E$. 
  We also show an algorithm with running time $O^*(2^{n/4})$ to find a max-weight
  popular matching matching in $G = (A \cup B,E)$ on $n$ vertices.
\end{abstract}

\section{Introduction}
\label{intro}
Consider a matching problem in $G = (A \cup B,E)$ where each vertex has a strict ranking of its
neighbors. The goal is to find an optimal way of pairing up vertices: stability is the usual notion of optimality in such
a setting. A matching $M$ is stable if $M$ admits no {\em blocking edge}, i.e., an edge $(a,b)$
such that $a$ and $b$ prefer each other to their respective assignments in $M$.
Stable matchings always exist in $G$ and can be computed in linear time~\cite{GS62}. 

In applications such as matching students to advisers or applicants to training posts, we
would like to replace the notion of ``no blocking edges'' with a more relaxed notion of 
``global stability'' for the sake of obtaining a matching that is {\em better} in a social sense,
for instance, a matching of larger size.
For this, we need to formalize the notion of a globally
stable matching; roughly speaking, a globally stable matching should be one such that there is no matching
where more people are happier.

This is precisely the notion of {\em popularity} introduced by G\"ardenfors~\cite{Gar75} in 1975. 
We say a vertex $u \in A \cup B$ {\em prefers} matching $M$ to matching $M'$ if either (i)~$u$ is matched in $M$
and unmatched in $M'$ or (ii)~$u$ is matched in both $M, M'$ and $u$ prefers $M(u)$ to $M'(u)$. 
For any two matchings $M$ and $M'$, let $\phi(M,M')$ be the number of vertices that prefer $M$ to $M'$.

\begin{definition}
\label{pop-def}
A matching $M$ is {\em popular} if  $\phi(M,M') \ge  \phi(M',M)$ for every matching $M'$ in $G$, 
i.e., $\Delta(M,M') \ge 0$ where $\Delta(M,M') = \phi(M,M') -  \phi(M',M)$.
\end{definition}

In an election between $M$ and $M'$ where vertices cast votes,  $\phi(M,M')$ is the number of votes 
for $M$ versus $M'$ and $\phi(M',M)$ is the number of votes for $M'$ versus $M$.  A popular matching 
never loses an election to another matching: thus it is a weak {\em Condorcet winner}~\cite{wiki-condorcet} 
in the corresponding voting instance. Although (weak) Condorcet winners need not exist in a general voting instance,
popular matchings always exist in a bipartite graph with strict preference lists, since every  stable matching is popular~\cite{Gar75}.

All stable matchings match the same subset of vertices~\cite{GS85} and the size of a stable matching can be as low as 
$|M_{\max}|/2$, where $M_{\max}$ is a max-size matching in $G$. 
One of the main motivations to relax stability to popularity is to obtain larger matchings and it is 
known that a max-size popular matching has size at least $2|M_{\max}|/3$. Polynomial time algorithms 
to compute a max-size popular matching in $G = (A \cup B,E)$ are known~\cite{HK11,Kav12}.

A {\em roommates} instance is a graph $G = (V,E)$, that is not necessarily bipartite,  with strict preference lists.
Stable matchings need not always exist in $G$ and there are
several polynomial time algorithms~\cite{Irv85,Sub94,TS98} to determine if a stable matching exists or not.
The definition of popularity (Definition~\ref{pop-def}) carries over to roommates instances --- though popular
matchings need not exist in $G$,  popular {\em mixed matchings}, i.e., probability distribution over matchings,
always exist in $G$ and can be efficiently computed~\cite{KMN09}. Observe that stable mixed matchings need not always 
exist in a roommates instance (see the Appendix).

We currently do not know the complexity of the {\em popular matching problem} in a roommates instance, i.e., 
does $G$ admit a popular matching?  The complexity of finding a {\em max-size} popular matching in $G$ 
was also an open problem so far and we show its hardness here.

\begin{theorem}
\label{thm3:roommates}
  The max-size popular matching problem in a roommates instance $G = (V,E)$ with strict preference lists is $\mathsf{NP}$-hard. 
\end{theorem}

We show the above problem is  $\mathsf{NP}$-hard even in instances that admit stable matchings.
Note that a stable matching is a min-size popular matching~\cite{HK11}.
All the polynomial time algorithms that compute popular matchings in bipartite graphs~\cite{GS62,HK11,Kav12} compute either 
stable matchings or {\em dominant} matchings. A popular matching $M$ is {\em dominant} if $M$ is more popular than every 
larger matching~\cite{CK16}, thus $M$ is a max-size popular matching.

Though the name ``dominant'' was given to this class of matchings in \cite{CK16},
dominant matchings in bipartite graphs first appeared in \cite{HK11}, which gave the first polynomial time
max-size popular matching algorithm in bipartite graphs.
More precisely, a matching satisfying Definition~\ref{def:strong-dominant} was constructed in the given bipartite graph. 
For any matching $M$, call an edge $(u,v)$ {\em negative} to $M$ if both $u$ and $v$ prefer their 
assignments in $M$ to each other.

\begin{definition}
  \label{def:strong-dominant}
A matching $M$ is strongly dominant in $G = (V,E)$ if there is a partition $(L,R)$ of the vertex set $V$ such that 
(i)~$M \subseteq L \times R$, (ii)~$M$ matches all vertices in $R$,
(iii)~every blocking edge to $M$ is in $R \times R$, and (iv)~every edge in $L \times L$ is negative to $M$.
\end{definition}

Consider the complete graph on 4 vertices $a, b, c, d$ where $a$'s preference list is $b \succ c \succ d$ (i.e., top choice $b$,
followed by $c$ and then $d$), $b$'s preference list is $c \succ a \succ d$, $c$'s preference list is $a \succ b \succ d$, and
$d$'s preference list is $a \succ b \succ c$. This instance has no stable matching. 
$M_1 = \{(a,d),(b,c)\}$ and $M_2 = \{(a,c),(b,d)\}$ are two strongly dominant matchings here: the corresponding partitions are
$L_1 = \{b,d\}, R_1 = \{a,c\}$ and $L_2 = \{c,d\}, R_2 = \{a,b\}$.

A strongly dominant matching $M$ with vertex partition $(L,R)$ is an $R$-perfect {\em stable matching} in the bipartite graph 
with $L$ on the left, $R$ on the right and edge set $E \cap (L \times R)$. It is also important to note that any blocking edge 
to $M$ in $G$ is in $R \times R$ and all edges in $L \times L$ are negative to $M$.

It was shown in \cite{CK16} that a popular matching $M$ is dominant if and only if there is {\em no} augmenting path with respect 
to $M$ in the subgraph obtained by deleting all negative edges with respect to $M$. In
bipartite graphs, dominant matchings and strongly dominant matchings are equivalent~\cite{CK16}. Moreover,
such a matching always exists in a bipartite graph and can be computed in linear time~\cite{Kav12}.

It can also be shown~\cite{Kav12} that every strongly dominant matching in $G=(V,E)$ is dominant.
However in non-bipartite graphs, not every dominant matching is strongly dominant. 
The complexity of the dominant matching problem in a roommates instance is currently not known.
Here we efficiently solve the strongly dominant matching problem.

\begin{theorem}
  \label{thm:strongly-dom}
  There is a linear time algorithm to determine if an instance $G = (V,E)$ with strict preference lists admits
  a strongly dominant matching or not and  if so, return one.
\end{theorem}

\subsection{Bipartite instances}
A natural generalization of the max-size  popular matching problem in a bipartite instance $G = (A \cup B, E)$
is the {\em max-weight} popular matching problem, where there is a weight function
$w: E \rightarrow \mathbb{R}$ and we seek a popular matching of largest weight. 
Several natural popular matching problems can be formulated with the help of edge weights: these include
computing a popular matching with as many of our ``favorite edges'' as possible or
an {\em egalitarian} popular matching (one that minimizes the sum of ranks of partners of all vertices).
Thus a max-weight popular matching problem is a generic problem that captures several optimization problems in popular
matchings. 

The max-weight {\em stable} matching problem is well-studied and there are several polynomial time 
algorithms~\cite{Feder92,Fed94,Fle03,ILG87,Rot92,TS98,VV89} to compute such
a matching or its variants in a bipartite graph with strict preference lists.
We show the following result here.

\begin{theorem}
  \label{thm:hard}
  The max-weight popular matching problem in $G = (A \cup B,E)$ with strict preference lists and a weight function
  $w: E \rightarrow \{1,2\}$ is $\mathsf{NP}$-hard.
\end{theorem}

A 2-approximate  max-weight popular matching in $G = (A \cup B,E)$ with strict preference lists and
non-negative edge weights can be computed in polynomial time.
We also show the following fast exponential time algorithm, where $n = |A \cup B|$.

\begin{theorem}
  \label{thm:exp-alg}
  A max-weight popular matching in $G = (A \cup B,E)$ with  strict preference lists and a 
weight function $w: E \rightarrow \mathbb{R}$ can be computed in $O^*(c^n)$ time, 
where $c = 2^{1/4} \approx 1.19$.
\end{theorem}

\subsection{Background and Related results}
Algorithmic questions for popular matchings were 
first studied in the domain of {\em one-sided} preference lists~\cite{AIKM07} in a bipartite instance $G = (A \cup B,E)$
where it is only vertices in $A$ that have preferences over their neighbors. Popular matchings need not always 
exist here, however popular mixed matchings always exist~\cite{KMN09}.
This proof extends to the domain of two-sided preference lists (with ties) and to non-bipartite graphs.

Popular matchings always exist in $G = (A \cup B,E)$ with two-sided 
strict preference lists. An $O(mn_0)$ algorithm to compute a max-size popular matching here was shown in \cite{HK11},
where $m = |E|$ and $n_0 = \min(|A|,|B|)$. A linear time algorithm for the max-size popular matching problem in such an
instance $G$ was shown in \cite{Kav12}.

A linear time algorithm was shown in \cite{CK16} to determine if there is a popular matching in $G = (A \cup B,E)$ that contains 
a given edge $e$. It was also shown in \cite{CK16} that
dominant matchings in $G = (A \cup B, E)$ are equivalent to stable matchings in a larger bipartite graph.
This equivalence implies a polynomial time algorithm to solve
the max-weight popular matching problem in a complete bipartite graph.

A description of the popular half-integral matching polytope of $G=(A\cup B, E)$ with strict preference lists
was given in \cite{Kav16}. It was shown in \cite{HK17} that the popular fractional matching polytope (from \cite{KMN09}) 
for such an instance $G = (A \cup B,E)$ is half-integral.
The half-integrality of the popular fractional matching polytope also holds for roommates instances~\cite{HK17}.

When preference lists admit ties, the problem of determining if a bipartite instance admits a popular 
matching or not is $\mathsf{NP}$-hard~\cite{BIM10,CHK15}. 
It is $\mathsf{NP}$-hard to compute a {\em least unpopularity factor} matching in a roommates instance~\cite{HK13}.
It was shown in \cite{HK17} that it is $\mathsf{NP}$-hard to compute a max-weight popular matching problem 
in a roommates instance with strict preference lists and it is $\mathsf{UGC}$-hard to compute a 
$\Theta(1)$-approximation.

The complexity of finding a max-weight popular matching in a bipartite instance 
with strict preference lists was left as an open problem in \cite{HK17}.
This problem along with the
complexity of finding a max-size popular matching in a roommates instance are two of the three open problems
in popular matchings listed in \cite{SocialChoice17} and we answer these two questions here.

\subsection{Techniques} 
Our results are based on LP-duality.
Every popular matching $M$ in an instance $G = (V,E)$ is a {\em max-cost} perfect matching in the graph $G$ with self-loops 
added and with edge costs given by a function $\wt_M$ (these costs depend on the matching $M$). 
Any optimal solution to the dual LP will be called a {\em witness} to $M$'s popularity.

\medskip

\noindent{\em Our hardness results.}
Witnesses for popular matchings in bipartite graphs first appeared in \cite{KMN09} and they were used in \cite{Kav16,HK17,BK17}. 
Roughly speaking, these algorithms dealt with matchings that had an element in $\{\pm 1\}^n$ as a witness.
Note that a stable matching has $0^n$ as a witness. For general popular matchings, there is no 
such ``parity agreement'' among the coordinates of any witness and we use this to show that 
the max-weight popular matching problem in bipartite graphs is $\mathsf{NP}$-hard.

All max-size popular matchings in a bipartite instance match the same subset of vertices~\cite{Hirakawa-MatchUp15}, 
however the {\em rural hospitals theorem} does not necessarily hold for max-size popular matchings in roommates instances.
Such an instance forms the main gadget in the proof of $\mathsf{NP}$-hardness for max-size popular
matchings in a roommates instance.

\medskip

\noindent{\em Our algorithms.}
We generalize the max-size popular matching algorithm for bipartite graphs~\cite{Kav12} to solve the strongly dominant matching 
problem in all graphs. We show a surprisingly simple reduction from the strongly dominant matching
problem in $G = (V,E)$ to the stable matching problem in a new roommates instance $G' = (V,E')$. Thus Irving's stable matching
algorithm~\cite{Irv85} in $G'$ solves our problem in linear time.

Our reduction is similar to an analogous correspondence in the bipartite case from \cite{CK16}.
However the graph $G'$ in \cite{CK16} on $3|A|+|B|$ vertices is {\em asymmetric} with respect to vertices in $A$ and $B$
of $G = (A \cup B,E)$. Now our new graph $G'$ may be regarded as the {\em bidirected} version of $G$,
i.e., each edge $(u,v)$ in $G$ is replaced by two edges $(u^+,v^-)$ and $(u^-,v^+)$ in $G'$.

Our fast exponential time algorithm for max-weight popular matchings in  $G = (A \cup B,E)$
formulates the convex hull ${\cal P}(\vec{r})$ of all popular matchings with at least 1
witness whose parities agree with a given vector $\vec{r} \in \{0,1\}^k$. Here $k$ is the number of components 
of size 4 or more in a subgraph $F_G$, whose edge set is the union of all popular matchings in $G$.

Our formulation of ${\cal P}(\vec{r})$ is based on the popular fractional matching polytope ${\cal P}_G$~\cite{KMN09}. 
We  use $\vec{r}$ to tighten several of the constraints in the formulation of ${\cal P}_G$ and
introduce new variables sandwiched between 0 and 1 to denote {\em fractional} parities and formulate a polytope. 
We use methods from \cite{HK17,TS98} along with some new ideas 
to show that our polytope is integral, more precisely, it is ${\cal P}(\vec{r})$. This leads to the $O^*(2^k)$ (where $k \le n/4$) 
algorithm.

\medskip

\noindent{\em Organization of the paper.}
Witness vectors for popular matchings in bipartite and roommates instances are defined in Section~\ref{prelims}.
Our algorithm for the strongly dominant matching problem in a roommates instance $G = (V,E)$ is given in Section~\ref{section4}.
Section~\ref{section3} shows that finding a max-size popular matching in $G$ is $\mathsf{NP}$-hard.
Section~\ref{section5} shows the $\mathsf{NP}$-hardness of the max-weight popular matching problem in a bipartite instance 
$G = (A \cup B, E)$. Our fast exponential time algorithm for this problem is given in  Section~\ref{section6}.

\section{Witness of a popular matching}
\label{prelims}

Let $M$ be any matching in our input instance $G = (V,E)$. In order to determine if 
$M$ is popular or not, we need to check if $\Delta(N,M) \le 0$ for all matchings $N$ in $G$ 
(see Definition~\ref{pop-def}). Computing $\max_N \Delta(N,M)$ reduces to computing a max-cost 
perfect matching in a graph $\tilde{G}$ with edge costs that are defined below.

The graph $\tilde{G}$ is the graph $G$ augmented with {\em self-loops} --- we assume that every vertex
is at the bottom of its own preference list. Adding self-loops allows us to view any matching $M$ in $G$
as a {\em perfect matching} $\tilde{M}$ in $\tilde{G}$ by including self-loops for all vertices left unmatched in $M$.
We now define a function $\wt_M$ on the edge set of $\tilde{G}$. For any edge $(u,v) \in E$, define:
\begin{equation*} 
\wt_M(u,v) = \begin{cases} 2   & \text{if\ $(u,v)$\ is\ a\ blocking\ edge\ to\ $\tilde{M}$}\\
	                     -2 &  \text{if\ $(u,v)$\ is\ a\ {\em negative}\ edge\ to\ $\tilde{M}$ }\\			
                              0 & \text{otherwise}
\end{cases}
\end{equation*}

Recall that an edge $(u,v)$ is {\em negative} to $\tilde{M}$ if both $u$ and $v$ prefer their partners in $\tilde{M}$ to each other.
Thus $\wt_M(u,v)$ is the sum of {\em votes} of $u$ and $v$ for each other over $\tilde{M}(u)$ and $\tilde{M}(v)$, 
respectively, where for any vertex $u$ and neighbors $v,v'$ of $u$: $u$'s vote for $v$ versus $v'$ is 1 if $u$ prefers $v$ to
$v'$, it is $-1$ if $u$ prefers $v'$ to $v$, else it is 0 (i.e. $v = v'$).
Observe that $\wt_M(u,v)= 0$ for any edge $(u,v) \in M$.

We now define $\wt_M$ for self-loops as well. For any $u \in V$, $\wt_M(u,u) = 0$ if $u$ is unmatched in $M$, else
$\wt_M(u,u) = -1$. Thus $\wt_M(u,u)$ is $u$'s vote for itself versus $\tilde{M}(u)$.

\begin{new-claim}
\label{claim1}
For any matching $N$ in $G$, $\Delta(N,M) = \wt_M(\tilde{N})$. 
\end{new-claim}
\begin{proof}
We will use the function $\vote(\cdot,\cdot)$ here. For any vertex $u$ and neighbors $v,v'$ of $u$ in $\tilde{G}$, 
recall that $\vote_u(v,v')$ is 1 if $u$ prefers $v$ to $v'$, it is -1 if $u$ prefers $v'$ to $v$, it is 0 otherwise (i.e., $v = v'$).

\smallskip 
 
Let $\tilde{M}(u)$ be $u$'s partner in $\tilde{M}$. Observe that
$\wt_M(a,b) = \vote_a(b,\tilde{M}(a)) + \vote_b(a,\tilde{M}(b))$. Also,
$\wt_M(u,u) = \vote_u(u,\tilde{M}(u))$.
We have the following equality from the definitions of $\Delta(\cdot,\cdot)$ and $\vote(\cdot,\cdot)$.
\begin{eqnarray*}
\Delta(N,M) & = & \sum_{u \in A\cup B}\vote_u(\tilde{N}(u),\tilde{M}(u))\\ 
            & = & \sum_{(a,b) \in N} (\vote_a(b,\tilde{M}(a)) + \vote_b(a,\tilde{M}(b))) \ + \sum_{(u,u) \in \tilde{N}}\vote_u(u,\tilde{M}(u)).
\end{eqnarray*}
This sum is exactly $\sum_{(a,b) \in N}\wt_M(a,b) +  \sum_{(u,u) \in \tilde{N}}\wt_M(u,u)$,
which is $\wt_M(\tilde{N})$. \qed
\end{proof}

Thus $M$ is popular if and only if every perfect matching in $\tilde{G}$ has cost at most 0.

\subsection{Bipartite instances}
Let $G = (A \cup B,E)$ be a bipartite instance with strict preference lists.
Consider the max-cost perfect matching LP in $\tilde{G}$: this is LP1 given below.
The set $\tilde{E}$ is the edge set of $\tilde{G}$ and $\tilde{E}(u)$ is the set of edges incident to $u$ in $\tilde{G}$. 
The linear program LP2 is the dual of LP1.

\begin{table}[ht]
\begin{minipage}[b]{0.45\linewidth}\centering
\begin{eqnarray*}
 \max \sum_{e \in \tilde{E}} \wt_M(e)\cdot x_e  &\mbox{\hspace*{0.1in}}\text{(LP1)}\\
      \text{s.t.}\qquad\sum_{e \in \tilde{E}(u)}x_e = 1  &\mbox{\hspace*{0.1in}}\forall\, u \in A \cup B\\
                        x_e  \ge 0   &\mbox{\hspace*{0.1in}}\forall\, e \in \tilde{E}. 
\end{eqnarray*}
\end{minipage}
\hspace{0.6cm}
\begin{minipage}[b]{0.45\linewidth}\centering
  \begin{eqnarray*}
\min \sum_{u \in A \cup B}\alpha_u  &\mbox{\hspace*{0.1in}}\text{(LP2)}\\
       \text{s.t.}\qquad \alpha_{a} + \alpha_{b} \ge \wt_{M}(a,b)  &\mbox{\hspace*{0.1in}}\forall\, (a,b)\in E\\
                          \alpha_u \ge \wt_M(u,u) &\mbox{\hspace*{0.1in}}\forall\, u \in A \cup B.\\
\end{eqnarray*}
\end{minipage}
\end{table}

$M$ is popular if and only if the optimal value of LP1 is at most 0 (by Claim~\ref{claim1}); in fact, the optimal value is exactly~0
since $\tilde{M}$ is a perfect matching in $\tilde{G}$ and $\wt_M(\tilde{M}) = 0$.
Thus $M$ is popular if and only if the optimal value of LP2 is 0 (by LP-duality).

\begin{definition}
  For any popular matching $M$, an optimal solution $\vec{\alpha}$ to LP2 above 
  is called a {\em witness} of $M$.
\end{definition}

A popular matching $M$ may have several witnesses. For any witness $\vec{\alpha}$, 
$\sum_{u \in A \cup B} \alpha_u =  0$ since $\vec{\alpha}$ is an optimal solution to LP2 and the optimal value of LP2 is~0.
Let $n = |A \cup B|$.

\begin{lemma}[\cite{Kav16}]
  \label{lem:witness}
Every popular matching $M$ in $G = (A \cup B,E)$ has a witness in $\{0,\pm 1\}^n$.
\end{lemma}

\subsection{Roommates instances}
Here our input is a graph $G = (V,E)$ with strict preference lists.
As done earlier, we will formulate the max-cost perfect matching problem in $\tilde{G}$ with cost function $\wt_M$
as our primal LP. 

The dual LP (\ref{LP3} given below) will be useful to us. Here $\Omega$ is
the collection of all odd subsets $S$ of $V$ of size at least 3 and $E[S]$ is the set of all $(u,v) \in E$ such that $u,v \in S$.
A matching $M$ is popular in $G = (V,E)$ if and only if the optimal value of \ref{LP3} is 0.

\begin{linearprogram}
  {
    \label{LP3}
    \minimize{\sum_{u \in V} \alpha_u \ + \ \sum_{S\in\Omega}\lfloor\,|S|/2\,\rfloor\cdot z_S}
  }
  \textstyle \alpha_u + \alpha_v + \sum_{\substack{S\in\Omega\\u,v\in S}}z_S \ & \ge \ \ \wt_M(u,v) \ \ \forall\, (u,v) \in E\notag  \\
  z_S \ \ge \ \ 0 \ \ \forall S \in \Omega\ \ \ \ & \mathrm{and} \ \ \ \ \alpha_u \  \ge \ \ \wt_M(u,u)\ \ \forall\, u \in V\notag.
\end{linearprogram}

\begin{definition}
For any popular matching $M$, an optimal solution $(\vec{\alpha}, \vec{z})$ to \ref{LP3} is called a {\em witness} of $M$.
\end{definition}

For any witness $(\vec{\alpha}, \vec{z})$, we have $\sum_{u \in V} \alpha_u + \sum_{S \in \Omega}\lfloor |S|/2\rfloor\cdot z_S = 0$.
Note that any stable matching in $G$ has $(\vec{0}, \vec{0})$ as witness.

Theorem~\ref{lem:strongly-dominant} gives a characterization of strongly dominant matchings
in terms of a special witness $(\vec{\alpha},\vec{z})$.
We will use this characterization of strongly dominant matchings in Section~\ref{section4}.

\begin{theorem}
  \label{lem:strongly-dominant}
  A matching $M$ is strongly dominant in $G$ if and only if there exists a feasible solution $(\vec{\alpha},\vec{z})$ to 
  \ref{LP3} such that $\alpha_u = \pm 1$ for all vertices $u$ matched in $M$, $\alpha_u = 0$ for all $u$ unmatched in~$M$, $z_S = 0$ 
  for all $S \in \Omega$, and $\sum_{u \in V}\alpha_u = 0$. 
\end{theorem}  
\begin{proof}
Let $M$ be a strongly dominant matching in $G = (V,E)$. So $V$ can be partitioned into $L \cup R$ such that properties~(i)-(iv) in
Definition~\ref{def:strong-dominant} are satisfied. Set $z_S = 0$ for all $S \in \Omega$.
We will now construct $\vec{\alpha}$ as follows. For $u \in V$:
\begin{itemize}
\item if $u \in R$ then set $\alpha_u = 1$
\item else if $u$ is matched in $M$ then set $\alpha_u = -1$ else set $\alpha_u = 0$.
\end{itemize}

Since $M$ matches all vertices in $R$, all vertices unmatched in $M$ are in $L$. Thus $\alpha_u = 0$ for all $u$ unmatched in $M$
and $\alpha_u = \pm 1$ for all $u$ matched in $M$. For any edge $(u,v) \in M$, since $M \subseteq L \times R$, 
$(\alpha_u,\alpha_v) \in \{(1,-1),(-1,1)\}$ and so $\alpha_u + \alpha_v = 0$. Thus $\sum_{u \in V}\alpha_u = 0$.

We will now show that $(\vec{\alpha},\vec{0})$ satisfies the constraints of \ref{LP3}. We have $\alpha_u \ge \wt_M(u,u)$.
This is because $\alpha_u = 0 = \wt_M(u,u)$ for $u$ left unmatched in $M$ and $\alpha_u \ge -1 = \wt_M(u,u)$ for $u$ matched in $M$. 
We will now show that all edge covering constraints are obeyed.
\begin{itemize}
\item Since $\wt_M(e) \le 2$ for any edge $e$ and $\alpha_u = 1$ for all $u \in R$, all edges in $R \times R$ are covered. 
\item We also know that any edge in $L \times L$ is a {\em negative} edge to $M$, i.e., $\wt_M(u,v) = -2$ for any edge 
$(u,v) \in L \times L$. Since $\alpha_u \ge -1$ for any $u \in L$, edges in $L \times L$ are covered. 
\item We also know that all blocking edges to $M$ are in $R \times R$ and so $\wt_M(u,v) \le 0$ for all 
$(u,v) \in L \times R$. Since $\alpha_u \ge -1$ and $\alpha_v = 1$,  all edges in $L \times R$ are covered.
\end{itemize}

We will now show the converse. Let $M$ be a matching with a witness $(\vec{\alpha},\vec{0})$ as given in the statement of the theorem.
To begin with, $M$ is popular since the objective function of \ref{LP3} evaluates to 0. We will now show that $M$ is strongly dominant.

For that, we will obtain a partition $(L, R)$ of $V$ as follows:
$R = \{u: \alpha_u = 1\}$ and $L = \{u: \alpha_u \ \text{is\ either}\ 0 \ \text{or -1}\}$.
Since $\tilde{M}$ is an optimal solution of the max-cost perfect matching problem in $\tilde{G}$,
complementary slackness conditions imply that if $(u,v) \in M$ 
then $\alpha_u + \alpha_v = \wt_M(u,v) = 0$. Since $u,v$ are matched, $\alpha_u,\alpha_v \in \{\pm 1\}$; so one of $u,v$ is in $L$ and the
other is in $R$. Thus $M \subseteq L \times R$.

We have $\wt_M(u,v) \le \alpha_u + \alpha_v$ for every $(u,v) \in E$. There cannot be any edge between 2 vertices left unmatched in
$M$ as that would contradict $M$'s popularity. So $\wt_M(u,v) \le -1$ for all 
$(u,v) \in E \cap (L \times L)$. Since $\wt_M(u,v) \in \{0,\pm 2\}$, $\wt_M(u,v) = -2$ for all edges 
$(u,v)$ in $L \times L$. In other words, every edge in $L \times L$ is negative to $M$. 

Moreover, any blocking edge can be present only
in $R \times R$ since $\wt_M(u,v) \le 1$ for all edges $(u,v) \in L \times R$. Finally, since $\alpha_u = \wt_M(u,u) = 0$ for all $u$ 
unmatched in $M$ (by complementary slackness conditions on \ref{LP3})
and every vertex $u \in R$ satisfies $\alpha_u = 1$, it means that all vertices in $R$ are matched in $M$. \qed
\end{proof}

\section{Strongly dominant matchings}
\label{section4}
In this section we show an algorithm to determine if a roommates instance $G = (V,E)$ admits a strongly dominant matching or not.
We will build a new roommates instance $G' = (V,E')$ and show that any stable matching in $G'$ can be projected to a strongly dominant
matching in $G$ and conversely, any strongly dominant matching in $G=(V,E)$ can be mapped to a stable matching in $G'$.

The vertex set of $G'$ is the same as that of $G$. Though there is only one copy of each vertex $u$ in $G'$, for every $(u,v) \in E$,
there will be 2 parallel edges in $G'$ between $u$ and $v$;
we will call one of these edges $(u^+,v^-)$ and the other $(u^-,v^+)$.
Thus $E' = \{(u^+,v^-), (u^-,v^+): \ (u,v) \in E\}$.
A vertex $v$ appears in 2 forms, as $v^+$ and $v^-$, to each of its neighbors.

We will now define preference lists in $G'$.
For any $u \in V$, if $u$'s preference list in $G$ is $v_1 \succ v_2 \succ \cdots \succ v_k$ then $u$'s
preference list in $G'$ is $v^-_1 \succ v^-_2 \succ \cdots \succ v^-_k \succ v^+_1 \succ v^+_2 \succ \cdots \succ v^+_k$.
Thus $u$ prefers any neighbor in ``$-$ form'' to any neighbor in ``$+$ form''.

As an example, consider the roommates instance on 4 vertices $a, b, c, d$ described in Section~\ref{intro}, where $d$ was the least
preferred vertex of $a, b, c$. Preference lists in the instance $G'$ are as follows:
\begin{minipage}[c]{0.45\textwidth}
			
			\centering
			\begin{align*}
				\ \ \ \ \ & a\colon  \, b^- \succ c^- \succ d^- \succ  b^+ \succ c^+ \succ d^+ \qquad\qquad\qquad && b\colon  \, c^- \succ a^- \succ d^- \succ  c^+ \succ a^+ \succ d^+\\
				\ \ \ \ \ & c\colon  \, a^- \succ b^- \succ d^- \succ  a^+ \succ b^+ \succ d^+ \qquad\qquad\qquad && d\colon  \, a^- \succ b^- \succ c^- \succ  a^+ \succ b^+ \succ c^+
			\end{align*}
\end{minipage}

\begin{itemize}
\item A matching $M'$ in $G'$ is a subset of $E'$ such that for each $u \in V$, $M'$ contains at most one edge incident to $u$, i.e.,
     at most one edge in 
     $\{(u^+,v^-),(u^-,v^+): v \in \Nbr(u)\}$ is in $M'$, where $\Nbr(u)$ is the set of $u$'s neighbors in $G$.
\item  For any matching $M'$ in $G'$, define the {\em projection} $M$ of $M'$ as follows:

\[M = \{(u,v): (u^+,v^-)\ \mathrm{or}\ (u^-,v^+)\ \mathrm{is\ in}\ M'\}.\]

It is easy to see that $M$ is a matching in $G$.
\end{itemize}

\begin{definition}
  A matching $M'$ is stable in $G'$ if for every edge $(u^+,v^-) \in E'\setminus M'$: either (i)~$u$ is matched in $M'$
  to a neighbor ranked better than $v^-$ or (ii)~$v$ is matched in $M'$ to a neighbor ranked better than $u^+$.
\end{definition}

We now present our algorithm to find a strongly dominant matching in $G = (V,E)$. 

\begin{enumerate}
\item Build the corresponding roommates instance $G' = (V,E')$.
\item Run Irving's stable matching algorithm in $G'$.
\item If a stable matching $M'$ is found in $G'$ then
  return the projection $M$ of $M'$.
  
  Else return ``$G$ has no strongly dominant matching''.
\end{enumerate}

In Irving's algorithm in $G'$, it is possible that a vertex $u$ proposes to some neighbor $v_t$ {\em twice}:
first to $v_t^-$ and later to $v_t^+$. During $u$'s first proposal, $v_t$ receives a proposal from $u^+$
and during $u$'s second proposal, $v_t$ receives a proposal from $u^-$.
We describe Irving's stable matching algorithm~\cite{Irv85} along with an example in the Appendix.

We will now prove the correctness of the above algorithm. 
We will first show that if our algorithm returns a matching $M$, then $M$ is a strongly
dominant matching in $G$.

\begin{lemma}
  \label{lemma2:main}
If $M'$ is a stable matching in $G'$ then the projection of $M'$ is a strongly dominant matching in $G$.
\end{lemma}
\begin{proof}
 Let $M$ be the projection of $M'$. In order to show that $M$ is a strongly dominant matching in $G$,
 we will construct a witness $(\vec{\alpha},\vec{0})$ as given in Theorem~\ref{lem:strongly-dominant}.
That is, we will construct $\vec{\alpha}$ as shown below such that $(\vec{\alpha},\vec{0})$ is a feasible solution to \ref{LP3}
and $\sum_{u\in V} \alpha_u = 0$.

 Set $\alpha_u = 0$ for all vertices $u$ left unmatched in $M$. For each vertex $u$ matched in $M$:
  \begin{itemize}
  \item  if $(u^+,\ast) \in M'$ then set $\alpha_u = 1$; else set $\alpha_u = -1$.
  \end{itemize}

  Note that $\sum_{u\in V} \alpha_u = 0$ since for each edge $(a,b) \in M$, we have $\alpha_a + \alpha_b = 0$ and for each
  vertex $u$ that is unmatched in $M$, we have $\alpha_u = 0$ . We also have $\alpha_u \ge \wt_M(u,u)$ for all $u \in V$ since 
  (i)~$\alpha_u = 0 = \wt_M(u,u)$ for all $u$ left unmatched in $M$ and (ii)~$\alpha_u \ge -1 = \wt_M(u,u)$ for all $u$ matched in $M$.
  
  We will now  show that for every $(a,b) \in E$, 
  $\alpha_a + \alpha_b \ge \wt_M(a,b)$. Recall that $\wt_M(a,b)$ is the sum of votes of $a$ and $b$ for each other over their respective
assignments in $M$.

  \begin{enumerate}
    \item Suppose $(a^+,\ast) \in M'$. So $\alpha_a = 1$. We will consider 3 subcases here. 
      \begin{itemize}
      \item The first subcase is that
      $(b^+,\ast) \in M'$. So $\alpha_b = 1$. Since $\wt_M(a,b) \le~2$, it follows that
      $\alpha_a + \alpha_b = 2 \ge \wt_M(a,b)$.
      \item The second subcase is that $(b^-,\ast) \in M'$. So $\alpha_b = -1$. If $(a^+,b^-) \in M'$ then  
           $\wt_M(a,b) = 0 = \alpha_a + \alpha_b$. So assume $(a^+,c^-)$ and $(b^-,d^+)$ belong to $M'$.       
       Since $M'$ is stable, the edge $(a^+,b^-)$ does not block $M'$. Thus either
       (i)~$a$ prefers $c^-$ to $b^-$ or (ii)~$b$ prefers $d^+$ to $a^+$. Hence $\wt_M(a,b) \le 0$ and so
       $\alpha_a + \alpha_b = 0 \ge \wt_M(a,b)$.
       \item The third subcase is that $b$ is unmatched in $M$. So $\alpha_b = 0$. 
       Since $M'$ is stable, the edge $(a^+,b^-)$ does not
       block $M'$. Thus $a$ prefers its partner $c^-$ in $M'$ to $b^-$ and so $\wt_M(a,b) = 0 < \alpha_a + \alpha_b$.
     \end{itemize}
     \item Suppose  $(a^-,\ast) \in M$. There are 3 subcases here as before. The case where $(b^+,\ast) \in M$ is
       totally analogous to the case where $(a^+,\ast)$ and $(b^-,\ast)$ are in $M$. So we will consider the remaining 2 subcases
       here.
       \begin{itemize}
       \item The first subcase is that $(b^-,\ast) \in M'$. So $\alpha_b = -1$. Let $(a^-,c^+)$ and $(b^-,d^+)$ belong to $M'$.
       Since $M'$ is stable, the edge $(a^+,b^-)$ does not block $M'$. So $b$ prefers $d^+$ to $a^+$.
       Similarly, the edge $(a^-,b^+)$ does not block $M'$. Hence $a$ prefers $c^+$ to $b^+$.
       Thus {\em both} $a$ and $b$ prefer their respective partners in $M$ to each other, i.e., $\wt_M(a,b) = -2$. So we have
       $\alpha_a + \alpha_b = -2 = \wt_M(a,b)$.

       \item The second subcase is that $b$ is unmatched in $M$. Then the edge $(a^+,b^-)$ {\em blocks} $M'$
       since $a$ prefers $b^-$ to $c^+$ (for any neighbor $c$) and $b$ prefers to be matched to $a^+$ than be left unmatched.
       Since $M'$ is stable and has no blocking edge, this means that this subcase does not arise.
       \end{itemize}   
  \item Suppose $a$ is unmatched in $M$. Then $(b^+,\ast) \in M'$ (otherwise $(a^-,b^+)$ blocks $M'$); moreover,
       $b$ prefers its partner $d^-$ in $M'$ to $a^-$. So we have $\wt_M(a,b) = 0 < \alpha_a + \alpha_b$.
     
  \end{enumerate}
  Thus we always have $\alpha_a + \alpha_b \ge \wt_M(a,b)$ and hence $(\vec{\alpha},\vec{0})$ is a valid witness of $M$. 
  Since $\vec{\alpha}$ satisfies the conditions in Theorem~\ref{lem:strongly-dominant}, $M$ is a strongly dominant matching in $G$.
  \end{proof}  

We will now show that if $G'$ has no stable matching, then $G$ has no strongly dominant matching.
\begin{lemma}
  \label{lemma1:main}
  If $G$ admits a strongly dominant matching then $G'$ admits a stable matching.
\end{lemma}
\begin{proof}
  Let $M$ be a strongly dominant matching in $G = (V,E)$. Let $(\vec{\alpha},\vec{0})$ be a witness of $M$ as given in 
  Theorem~\ref{lem:strongly-dominant}.
  That is, $\alpha_u = 0$ for $u$ unmatched in $M$ and $\alpha_u = \pm 1$ for $u$ matched in $M$.
  Moreover, for each $(u,v) \in M$, $\alpha_u + \alpha_v = \wt_M(u,v) = 0$ by complementary slackness on \ref{LP3}; 
  so $(\alpha_u,\alpha_v)$ is either $(1,-1)$ or $(-1,1)$.

  We will construct a stable matching $M'$ in $G'$ as follows. For each $(u,v) \in M$: 
  \begin{itemize}
  \item if $(\alpha_u,\alpha_v) = (1,-1)$ then add $(u^+,v^-)$ to $M'$;
        else add $(u^-,v^+)$ to $M'$.
  \end{itemize}

We will show that no edge in $E'\setminus M'$ blocks $M'$.
Let $(a^+,b^-) \notin M'$. We consider the following cases here:

\smallskip
  
\noindent{\bf Case 1.} Suppose $\alpha_b = 1$. Then $(b^+,d^-) \in M'$ where $d = M(b)$. Since $b$ prefers $d^-$ to $a^+$, 
                         $(a^+,b^-)$ is not a blocking edge to $M'$.
\smallskip
  
\noindent{\bf Case 2.} Suppose $\alpha_b = -1$. Then  $(b^-,d^+) \in M'$ where $d = M(b)$.
  We have 2 sub-cases here: (i)~$\alpha_a = 1$ and (ii)~$\alpha_a = -1$.
  Note that $\alpha_a \ne 0$ as the edge $(a,b)$ would not be covered by $\alpha_a + \alpha_b$ then. This is because if $\alpha_a = 0$ then
  $a$ is unmatched in $M$ and $\wt_M(a,b) = 0$ while $\alpha_a + \alpha_b = -1$.
  \begin{itemize}
  \item In sub-case~(i), some edge $(a^+,c^-)$ belongs to $M'$. 
  We know that $\wt_M(a,b) \le \alpha_a + \alpha_b = 0$, so either
  (1)~$a$ prefers $M(a) = c$ to $b$ or (2)~$b$ prefers $M(b) = d$ to $a$. Hence either (1)~$a$ prefers $c^-$ to $b^-$ or
  (2)~$b$ prefers $d^+$ to $a^+$. Thus $(a^+,b^-)$ is not a blocking edge to $M'$.
  \item In sub-case~(ii), some edge $(a^-,c^+)$ belongs to $M'$. We know that $\wt_M(a,b) \le \alpha_a + \alpha_b = -2$, so
  $a$ prefers $M(a) = c$ to $b$ {\em and} $b$ prefers $M(b) = d$ to $a$. Thus $b$ prefers $d^+$ to $a^+$, hence  
  $(a^+,b^-)$ is not a blocking edge to $M'$.
  \end{itemize}
  
\noindent{\bf Case 3.} Suppose $\alpha_b = 0$. Thus $b$ was unmatched in $M$. Each of $b$'s neighbors has to be matched in $M$ to a
  neighbor that it prefers to $b$, otherwise $M$ would be unpopular. We have $\alpha_a + \alpha_b \ge \wt_M(a,b) = 0$, hence it follows
  that $\alpha_a = 1$. Thus $(a^+,c^-) \in M'$ where $c$ is a neighbor that $a$ prefers to $b$. So $(a^+,b^-)$ is not a blocking edge
  to $M'$. \qed
\end{proof}

Lemmas~\ref{lemma2:main} and \ref{lemma1:main} show that a strongly dominant matching is present in $G$ if and only if a stable
matching is present in $G'$. This finishes the proof of correctness of our algorithm.
Since Irving's stable matching algorithm in $G'$ can be implemented to run in linear time~\cite{Irv85}, we can conclude
Theorem~\ref{thm:strongly-dom} stated in Section~\ref{intro}.

\section{The max-size popular matching problem in a roommates instance}
\label{section3}
In this section we prove the $\mathsf{NP}$-hardness of the max-size popular matching problem in a roommates instance.
We will show a reduction from the vertex cover problem.

Let $H = (V_H,E_H)$ be an instance of the vertex cover problem and let $V_H = \{1,\ldots,n_H\}$, i.e., $V_H = [n_H]$. We will build a
roommates instance $G$ as follows: (see Fig.~\ref{newfig1:example})
\begin{itemize}
\item corresponding to every vertex $i \in V_H$, there will be  4 vertices $a_i,b_i,c_i,d_i$ in $G$  and
\item corresponding to every edge $e = (i,j) \in E_H$, there will be  2 vertices  $u^e_i$ and $u^e_j$ in $G$.
\end{itemize}

\begin{figure}[h]
\centerline{\resizebox{0.56\textwidth}{!}{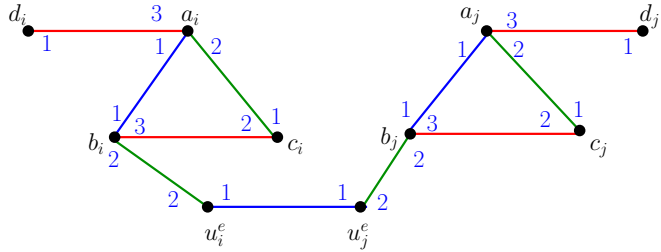}}
\caption{The graph $G$ restricted to the adjacent vertices $i$ and $j$ in $H$: the vertices $a_t,b_t,c_t,d_t$ in $G$ correspond to vertex $t \in \{i,j\}$ in $H$ and the vertices $u^e_i$ and $u^e_j$ in $G$ correspond to the edge $e = (i,j)$ in $H$. Vertex preferences in $G$ are indicated on the edges.}
\label{newfig1:example}
\end{figure}

The preferences of the vertices $a_i,b_i,c_i,d_i$ are as follows, where $e_1,\ldots,e_k$ are all the edges in $H$ with vertex $i$ as an 
endpoint.
\[a_i\colon \,b_i \succ c_i \succ d_i \ \ \ \ \ \ \ \ \ b_i\colon\,a_i \succ u^{e_1}_i\cdots\succ u^{e_k}_i \succ c_i\ \ \ \ \ \ \ \ \ c_i\colon \, a_i \succ b_i\ \ \ \ \ \ \ \ \ d_i\colon \, a_i.\]
The order among the vertices $u^{e_1}_i,\ldots,u^{e_k}_i$ in the preference list of $b_i$ is arbitrary.
The preference list of vertex $u^e_i$ is $u^e_j \succ b_i$, where $e = (i,j)$ (see Fig.~\ref{newfig1:example}). 

Observe that $G$ admits a stable matching $S = \{(a_i,b_i): 1 \le i \le n_H\} \cup \{(u^e_i, u^e_j): e=(i,j) \in E_H\}$.
This is the set of blue edges in Fig.~\ref{newfig1:example}.

\begin{lemma}
\label{hard:lemma1}
Let $M$ be a popular matching in $G$.
\begin{itemize}
\item  For any $i \in [n_H]$, either $(a_i,b_i) \in M$ or $\{(a_i,d_i),(b_i,c_i)\} \subseteq M$.
\item The set $C = \{i \in [n_H]: (a_i,b_i) \in M\}$ is a vertex cover of $H$.
\end{itemize}
\end{lemma}
\begin{proof}
  The vertex $a_i$ is the top choice neighbor of all its neighbors $b_i,c_i,d_i$. Thus $a_i$ has to be matched in the popular matching $M$.
  \begin{enumerate}
  \item If $a_i$ is matched to $b_i$ then $(a_i,b_i) \in M$.
  \item If $(a_i,d_i) \in M$, then $c_i$ also has to be matched --- otherwise we get a more popular matching by replacing the
  edge $(a_i,d_i)$ with $(a_i,c_i)$. Since $c_i$ has degree 2, it has to be the case that $\{(a_i,d_i),(b_i,c_i)\} \subseteq M$.
\item Suppose $(a_i,c_i) \in M$. Since $a_i$ prefers $b_i$ to $c_i$, this means that $b_i$ has to be matched in $M$.
      Other than $a_i$ and $c_i$ (which are matched to each other), $b_i$'s neighbors are $u^e_i$ for all edges $e$ incident to $i$ in $H$.
      So $(b_i,u^e_i) \in M$ for some $e = (i,j) \in E_H$.
    We will now construct a new matching $M'$ as follows:
    \begin{itemize}
    \item replace the edges $(a_i,c_i),(b_i,u^e_i),(u^e_j, M(u^e_j))$ in $M$ with the edges $(a_i,b_i)$ and $(u^e_i,u^e_j)$. So
      $c_i$ and $M(u^e_j)$ are unmatched in $M'$, hence these 2 vertices prefer $M$ to $M'$; however the 4 vertices $a_i,b_i,u^e_i,u^e_j$
      prefer $M'$ to $M$.
      Thus $M'$ is more popular than $M$, a contradiction to the popularity of $M$. Hence $(a_i,c_i) \notin M$.
    \end{itemize}
\end{enumerate}    

\smallskip

  We now show the second part of the lemma.
  Let $(i,j) \in E_H$. We need to show that either $(a_i,b_i) \in M$ or  $(a_j,b_j) \in M$. Suppose not. Then by the first part of 
this lemma,  $(a_i,d_i),(b_i,c_i)$ are in $M$ and similarly, $(a_j,d_j)$, $(b_j,c_j)$ are in $M$. Also  $(u^e_i,u^e_j) \in M$.

  Consider the matching $M'$ obtained by replacing the 5 edges $(a_i,d_i),(b_i,c_i),(a_j,d_j)$, $(b_j,c_j)$, and $(u^e_i,u^e_j)$ in $M$
  with the 4 edges $(a_i,c_i),(b_i,u^e_i),(a_j,c_j)$, and $(b_j,u^e_j)$ (see Fig.~\ref{newfig1:example}). Among the 10 vertices involved here,
  the 6 vertices $a_i,b_i,c_i$ and $a_j,b_j,c_j$ prefer $M'$ to $M$ while the 4 vertices $d_i,u^e_i$ and $d_j,u^e_j$ prefer $M$ to $M'$. Thus
  $M'$ is more popular than $M$, a contradiction to the popularity of $M$.

  Hence for each edge $(i,j) \in E_H$ either $(a_i,b_i) \in M$ or  $(a_j,b_j) \in M$. In other words, the set
  $U = \{i \in [n_H]: (a_i,b_i) \in M\}$ is a vertex cover of $H$. \qed
\end{proof}

\begin{theorem}
  \label{thm:max-size-hard}
  For any $1 \le k \le n_H$, the graph $H=(V_H,E_H)$ admits a vertex cover of size $k$ if and only if $G$ has a popular matching of size
  at least $m_H + 2n_H - k$, where $|E_H| = m_H$.
\end{theorem}
\begin{proof}
  Suppose $H=(V_H,E_H)$ admits a vertex cover $U$ of size $k$. We will now build a popular matching $M$ in $G$ of size $m_H + 2n_H - k$.
  \begin{itemize}
    \item  Add all edges $(u^e_i,u^e_j)$ in $G$ to $M$.
    \item For every $i \in U$, add the edge $(a_i,b_i)$ to $M$. 
    \item For every $i \notin U$, add the edges $(a_i,d_i)$, $(b_i,c_i)$ to $M$. 
  \end{itemize}

  The size of $M$ is $m_H + |U| + 2(n_H - |U|) = m_H + 2n_H - k$.
  We will prove $M$ is popular by showing a witness $(\vec{\alpha},\vec{z})$ for it. To begin with, initialize $z_S = 0$ for all 
$S \in \Omega$.
  \begin{itemize}
    \item For every $i \in U$: set $\alpha_{a_i} = \alpha_{b_i} = \alpha_{c_i} = \alpha_{d_i} = 0$.
    \item For every $i \notin U$: set $\alpha_{a_i} = 1$ and $\alpha_{b_i} = \alpha_{c_i} = \alpha_{d_i} = -1$;
                                  also set $z_{S_i} = 2$ where $S_i = \{a_i,b_i,c_i\}$.
    \item For every edge $e = (i,j) \in E_H$: if $i \in U$ then set $\alpha_{u^e_i} = -1$ and $\alpha_{u^e_j} = 1$; else set $\alpha_{u^e_i} = 1$ and $\alpha_{u^e_j} = -1$.
     
   \end{itemize}    
  It is easy to check that the above setting of $(\vec{\alpha},\vec{z})$ covers all edges of $G$.
  In particular, for $i \in U$, we have $\wt_M(b_i,u^e_i) = -1 -1 = -2$ while $\alpha_{b_i} = 0$ and $\alpha_{u^e_i} = -1$, thus
  $\alpha_{b_i} + \alpha_{u^e_i} \ge \wt_M(b_i,u^e_i)$. Similarly, when $i \notin U$, $\wt_M(b_i,u^e_i) = 1 - 1 = 0$ and we have
  $\alpha_{b_i} = -1$ and $\alpha_{u^e_i} = 1$ here. Moreover,
  \begin{equation}
  \label{new-eqn6}
  \sum_{v\in V}\alpha_v \ + \ \sum_{S\in\Omega}\lfloor |S|/2\rfloor\cdot z_S \ \ = \ \ \sum_{i \notin U}\,-2 \ + \ \sum_{i \notin U}\,2 \ \ = \ \ 0.
  \end{equation}

  This is because $\alpha_v = 0$ for all vertices $v$ unmatched in $M$ and $\alpha_u + \alpha_v = 0$ for all edges $(u,v) \in M$ except 
  the edges $(b_i,c_i)$ where $i \notin U$. 
  For each $i \notin U$, we have $\alpha_{b_i} + \alpha_{c_i} = -2$ and we also have $z_{S_i} = 2$ where $S_i = \{a_i,b_i,c_i\}$. 
  The sum in Equation~(\ref{new-eqn6}) is 0 and so $M$ is a popular matching.  Hence $G$ has a popular matching of size $m_H + 2n_H - k$.

  \medskip
  
  We will now show the converse.  Let $M$ be a popular matching in $G$ of size at least $m_H + 2n_H - k$.
  We know that $U = \{i \in [n_H]: (a_i,b_i) \in M\}$ is a vertex cover of $H$ (by Lemma~\ref{hard:lemma1}).
  We will show that $|U| \le k$.

  It follows from Lemma~\ref{hard:lemma1} that all edges $(u^e_i,u^e_j)$ belong to any popular matching. So these account for $m_H$ many edges in $M$.
  We also know that for every $i \in V_H$, either $(a_i,b_i) \in M$ or $\{(a_i,d_i),(b_i,c_i)\} \subseteq M$ (by Lemma~\ref{hard:lemma1}).
  Thus $|M| = m_H + |U| + 2(n_H - |U|)$. Since $|M| \ge m_H + 2n_H - k$, it follows that $|U| \le k$.
  Thus the graph $G$ has a vertex cover of size $k$. \qed
\end{proof}

We can now conclude Theorem~\ref{thm3:roommates} stated in Section~\ref{intro}.

\medskip

\noindent{\em Remark.} The rural hospitals theorem~\cite{Roth84} for stable matchings in a roommates instance $G$ says that every stable 
matching in $G$ matches the same subset of vertices. Such a statement is not true for max-size popular matchings in a roommates instance
as seen in the instance on 10 vertices in Fig.~\ref{newfig1:example}. This instance has two max-size popular matchings 
(these are of size 4): one leaves $c_i$ and $d_i$ unmatched while another leaves $c_j$ and $d_j$ unmatched.

\section{Max-weight popular matchings in bipartite instances}
\label{section5}
In this section  we will prove the $\mathsf{NP}$-hardness of the max-weight popular matching problem in a
bipartite instance $G = (A \cup B,E)$ on $n$ vertices and with edge weights.
Call any $e \in E$ a {\em popular edge} if there is a popular matching $N$ in $G$ such 
that $e \in N$. 

Let $M$ be a popular matching in $G = (A \cup B,E)$.
Observe that for any popular matching $N$ in $G$, the perfect matching $\tilde{N}$ in $\tilde{G}$ 
is an optimal solution of LP1 (see Section~\ref{prelims}) 
and a witness $\vec{\alpha}$ of $M$ is an optimal solution of LP2.
Lemma~\ref{prop0} follows from complementary slackness conditions on LP1.

\begin{lemma}
\label{prop0}
Let $M$ be a popular matching in $G$ and let $\vec{\alpha} \in\{0,\pm 1\}^n$ be a witness of $M$.
\begin{enumerate}
\item For any popular edge $(a,b) \in E$, the parities of $\alpha_a$ and $\alpha_b$ have to be the same.
\item If $u$ is a vertex in $G$ that is left unmatched in a stable matching in $G$ (call $u$ unstable) 
then $\alpha_u = \wt_M(u,u)$; thus $\alpha_u = 0$ if $u$ is left unmatched in $M$, otherwise $\alpha_u = -1$.
\end{enumerate}
\end{lemma}
\begin{proof}
  Let $(a,b)$ be a popular edge. So $(a,b) \in N$ for some popular  matching $N$.
  Since $N$ is popular, $\Delta(N,M) = 0$ and the perfect matching $\tilde{N}$ is an optimal solution to the
  max-cost perfect matching LP in the graph $\tilde{G}$ with cost function $\wt_M$ (see LP1 from Section~\ref{prelims}).
  Since $\vec{\alpha}$ is an optimal solution to the dual LP (see LP2), it follows from complementary slackness that 
$\alpha_a +  \alpha_b = \wt_M(a,b)$. Observe that $\wt_M(a,b) \in \{\pm 2, 0\}$ (an even number).
Hence the integers $\alpha_a$ and $\alpha_b$ have the same parity. This proves part~1. 

Part~2 also follows from complementary slackness. Let $S$ be a stable matching in $G$ and let $u$ be 
a vertex left unmatched in $S$. So the perfect matching $\tilde{S}$ 
contains the edge $(u,u)$. Since $\tilde{S}$ is an
optimal solution to LP1, we have $\alpha_u = \wt_M(u,u)$ by complementary slackness.
Thus when $u$ is left unmatched in~$M$, $\alpha_u = 0$,
else  $\alpha_u = -1$. \qed
\end{proof}

\noindent{\bf The popular subgraph.}
We will define a subgraph $F_G = (A \cup B, E_F)$ called the {\em popular subgraph} of $G$, where $E_F$ is the set of 
popular edges in $E$.
The subgraph $F_G$ need not be connected: let ${\cal C} = \{\Comp_1,\ldots,\Comp_h\}$ be the set of 
connected components in $F_G$. 

Each component $\Comp_j$ that is a singleton set consists of a single {\em unpopular} vertex, i.e., one left unmatched in all
popular matchings. Every non-singleton component $\Comp_i$ has an {\em even} number of vertices.
This is because a max-size popular matching matches all vertices except the ones in singleton 
sets in ${\cal C}$ and vertices in $\Comp_i$ are matched to each other~\cite{Hirakawa-MatchUp15}.

It is also known that all {\em stable} vertices (those matched in stable matchings in $G$)
have to be matched in every popular matching~\cite{HK11}.
Let $M$ be any popular matching in $G$ and let $\vec{\alpha} \in\{0,\pm 1\}^n$ be a witness of $M$'s popularity. 
The following lemma will be useful to us.
\begin{lemma}
\label{prop1}
For any connected component $\Comp_i$ in $F_G$, either $\alpha_u = 0$ for all vertices $u \in \Comp_i$ or 
$\alpha_u = \pm 1$ for all vertices $u \in \Comp_i$. Moreover, if $\Comp_i$ contains one or 
more unstable vertices, 
either all the unstable vertices in $\Comp_i$ are matched in $M$ or none of them is matched in $M$.
\end{lemma}
\begin{proof}
  Let $u$ and $v$ be any 2 vertices in $\Comp_i$.
  Since $u, v$ are in the same connected component in $F_G$, there is a $u$-$v$ path $\rho$ in $G$ such that 
every edge in $\rho$ is a {\em popular edge}. 
The endpoints of each popular edge have the same parity in $\vec{\alpha}$ (by part~1 of Lemma~\ref{prop0}), 
hence $\alpha_u$ and $\alpha_v$ have the same parity. Thus either 
$\alpha_u = \alpha_v = 0$ or both $\alpha_u, \alpha_v \in \{\pm 1\}$.

Let $\Comp_i$ be a connected component with one or more {\em unstable} vertices, i.e., those left unmatched in a stable matching. 
Since all $\alpha$-values in $\Comp_i$ have the same parity, either (i)~all vertices $v$ in $\Comp_i$ satisfy $\alpha_v = 0$ 
or (ii)~all vertices $v$ in $\Comp_i$ satisfy  $\alpha_v = \pm 1$. For any unstable vertex $u$,
we have $\alpha_u = \wt_M(u,u)$ (by part~2 of Lemma~\ref{prop0}), hence in case~(i), all unstable vertices in 
$\Comp_i$ are left unmatched in $M$ and in case~(ii), all unstable 
vertices in $\Comp_i$ are matched in $M$. \qed
\end{proof}

\medskip

\noindent{\bf The $\mathsf{NP}$-hardness reduction.}
Given a graph $H = (V_H,E_H)$ which is an instance of the vertex cover problem, we will now build an instance
$G = (A\cup B,E)$ with strict preference lists such that the vertex cover problem in $H$ reduces to the max-weight
popular matching problem in $G$. Let $V_H = \{1,\ldots,n_H\}$.

\begin{itemize}
\item For every edge $e \in E_H$, there will be a gadget $D_e$ in $G$ on 6 vertices $s_e,t_e,s'_e,t'_e,s''_e,t''_e$.
\item For every vertex $i \in V_H$, there will be a gadget $C_i$ in $G$ on 4 vertices $a_i,b_i,a'_i,b'_i$.
\item There are 2 more vertices in $G$: these are $a_0$ and $b_0$.
\end{itemize}

Thus $A = \{a_0\} \cup \{a_i,a'_i: i \in V_H\} \cup \{s_e,s'_e,s''_e: e \in E_H\}$ and
$B = \{b_0\} \cup \{b_i,b'_i: i \in V_H\}\cup$ $\{t_e,t'_e,t''_e: e \in E_H\}$.
We now describe the edge set of $G$.
We first describe the preference lists of the 6 vertices in $D_e$, where $e = (i,j)$ and $i < j$ (see Fig.~\ref{fig:hardness1}).

\begin{minipage}[c]{0.45\textwidth}
			
			\centering
			\begin{align*}
                          &  s'_e\colon \, t'_e \succ t_e  \qquad\qquad && s''_e\colon \, t''_e \succ t_e \qquad\qquad && s_e\colon \, t'_e \succ b_j \succ t''_e\\
                          &  t'_e\colon  \, s'_e \succ s_e \qquad\qquad && t''_e\colon \, s''_e \succ s_e \qquad\qquad && t_e\colon  \, s''_e \succ a_i \succ s'_e\\
			\end{align*}
\end{minipage}

Here $s'_e$ and $t'_e$ are each other's top choices and similarly, $s''_e$ and $t''_e$ are each 
other's top choices. The vertex $s_e$'s top choice is $t'_e$, second choice is $b_j$, and third 
choice is $t''_e$. For $t_e$, the order is $s''_e$, followed by $a_i$, and then $s'_e$.
Recall that $e = (i,j)$ and $i < j$.
\begin{figure}[h]
\centerline{\resizebox{0.65\textwidth}{!}{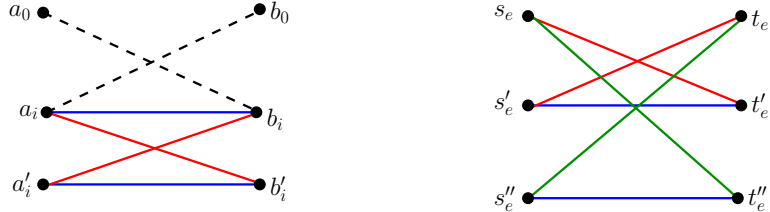}}
\caption{To the left is the gadget $C_i$ on vertices $a_i,b_i,a'_i,b'_i$ and to the right is the gadget $D_e$ on vertices $s_e,t_e,s'_e,t'_e,s''_e,t''_e$; the vertices $a_0$ and $b_0$ are adjacent to $b_i$ and $a_i$, respectively.}
\label{fig:hardness1}
\end{figure}

We now describe the preference lists of the 4 vertices $a_i,b_i,a'_i,b'_i$ in $C_i$ (see Fig.~\ref{fig:hardness1}).

\begin{minipage}[c]{0.45\textwidth}
			
			\centering
			\begin{align*}
				& a'_i\colon  \, b_i \succ b'_i \qquad\qquad && a_i\colon \, b_i \succ b'_i \succ b_0 \succ \cdots \\
				& b'_i\colon  \, a_i \succ a'_i \qquad\qquad && b_i\colon \, a_i \succ a'_i \succ a_0 \succ \cdots
			\end{align*}
\end{minipage}

\begin{itemize}
\item Both $a'_i$ and $a_i$ have $b_i$ as their top choice and $b'_i$ as their second choice. Similarly, 
both $b'_i$ and $b_i$ have $a_i$ as their top choice and $a'_i$ as their second choice.
  
\item The vertex $a_i$ has other neighbors: its third choice is $b_0$ followed by all the vertices 
$t_{e_1},t_{e_2},\ldots$  where $i$ is the {\em lower-indexed}
endpoint of $e_1,e_2,\ldots$. Similarly, $b_i$ has  $a_0$ as its third choice followed by all 
the vertices $s_{e'_1},s_{e'_2},\ldots$ where $i$ is the {\em higher-indexed} endpoint of $e'_1,e'_2,\ldots$

\item The order among the vertices $t_{e_1},t_{e_2},\ldots$ (similarly, 
$s_{e'_1},s_{e'_2},\ldots$) in the preference list of $a_i$ (resp., $b_i$) does not matter;
hence these are represented as $\cdots$ in these preference lists.
\end{itemize}

The vertex $a_0$ has $b_1,\ldots, b_{n_H}$ as its neighbors and its preference list is some arbitrary 
permutation of these vertices. Similarly, the vertex $b_0$ has $a_1,\ldots, a_{n_H}$ as its neighbors 
and its preference list is some arbitrary permutation of these vertices.

\begin{lemma}
\label{cor:redn2}  
No popular matching in $G$ matches either $a_0$ or $b_0$. 
\end{lemma}
\begin{proof}
  We will first show that $M = \{(a_i,b'_i),(a'_i,b_i): i \in V_H\}$ $\cup \{(s'_e,t'_e),(s''_e,t''_e): e \in E_H\}$
  is a popular matching in $G$. We prove the popularity of $M$ by showing a vector 
  $\vec{\alpha} \in \{0, \pm 1\}^n$ that will be a {\em witness}  to $M$'s popularity.

  \begin{itemize}
  \item For vertices in $C_i$, where $i \in V_H$: let $\alpha_{a_i} =  \alpha_{b_i} = 1$, 
        $\alpha_{a'_i} =  \alpha_{b'_i} = -1$.
  \item For any vertex $u$ in $D_e$, where $e \in E_H$: let $\alpha_u = 0$.
  \item Let $\alpha_{a_0} = \alpha_{b_0} = 0$.
  \end{itemize}

  It can be checked that $\alpha_u + \alpha_v \ge \wt_M(u,v)$ for all edges $(u,v)$ in $G$. 
  Also, $\alpha_u \ge \wt_M(u,u)$ for all $u \in A \cup B$. 
  Since $\alpha_u = 0$ for all vertices $u$ unmatched in $M$ and $\alpha_u + \alpha_v = 0$ for every edge $(u,v)$ in $M$,  
  $\sum_{u \in A\cup B}\alpha_u = 0$.  

 \smallskip
  
Suppose $a_0$ is matched in some popular matching $N$, i.e., $(a_0,b_i) \in N$ for some 
$i \in [n_H]$. Then $(a_0,b_i)$ is a {\em popular edge}. So it follows from 
part~1 of Lemma~\ref{prop0} that the parities of $\beta_{a_0}$ and $\beta_{b_i}$
have to be the same, for any witness $\vec{\beta}$ of the popular matching $M$.
However $\alpha_{a_0} = 0$ while $\alpha_{b_i} = 1$ for the witness $\vec{\alpha}$ described above.
This is a contradiction to $(a_0,b_i)$ being a popular edge and
hence no popular matching in $G$ matches $a_0$ (similarly, $b_0$). \qed
\end{proof}

\begin{lemma}
\label{lem:redn3}  
Let $N$ be any popular matching in $G$. For each $i \in \{1,\ldots,n_H\}$: either $(a_i,b_i)$ and $(a'_i,b'_i)$
are in $N$ or $(a_i,b'_i)$ and $(a'_i,b_i)$ are in $N$.
\end{lemma}
\begin{proof}
We know from Lemma~\ref{cor:redn2} that in the popular subgraph $F_G$, the vertices $a_0$ and 
$b_0$ are singleton sets. We now claim that each $C_i$ forms a maximal connected component in $F_G$.
This implies that $a_i$ has only 2 possible partners in any popular matching $N$: either $b_i$ or $b'_i$. 
Thus either 
(1)~$\{(a_i,b_i),(a'_i,b'_i)\} \subset N$ or (2)~$\{(a_i,b'_i),(a'_i,b_i)\} \subset N$.

\smallskip

It is easy to see that the 4 vertices of $C_i$ belong to the same connected component in $F_G$: 
this is because all the 4 edges $(a_i,b_i),(a'_i,b'_i),(a_i,b'_i)$ ,$(a'_i,b_i)$ are popular. The 
former 2 edges belong to the  stable matching $S = \{(a_i,b_i),(a'_i,b'_i): i \in V_H\}$ $\cup \{(s'_e,t'_e),(s''_e,t''_e): e \in E_H\}$
while the latter 2 edges belong to the popular matching $M$ defined in Lemma~\ref{cor:redn2}. 

Consider the matching $M$ and its witness $\vec{\alpha}$ defined in the proof of Lemma~\ref{cor:redn2}: 
the parities of $\alpha_{a_i}$ and $\alpha_{t_e}$ for any $i \in V_H$ and $e \in E_H$ are different.
Hence it follows from part~1 of Lemma~\ref{prop0} that the edge $(a_i,t_e)$ is {\em not} a popular 
edge. Similarly, the edge $(s_e,b_j)$ is not a popular edge. Thus each 
$C_i$  forms a maximal connected component in $F_G$. \qed
\end{proof}

\medskip

Lemma~\ref{lem:redn3} will be important in our reduction. 
Let $\vec{\alpha}$ be any witness of $N$. 
\begin{itemize}
\item In the first possibility of Lemma~\ref{lem:redn3}, i.e., when
$(a_i,b_i),(a'_i,b'_i)$ are in $N$, we have $\alpha_u = 0$ for all $u \in C_i$. This is 
because $a_0$ and $b_0$ are unmatched in $N$, so
$\alpha_{a_i} + \alpha_{b_0} \ge \wt_N(a_i,b_0) = 0$ and $\alpha_{a_0} + \alpha_{b_i} \ge \wt_N(a_0,b_i) = 0$.
Also $\alpha_{a_0} = \alpha_{b_0} = 0$ (by Lemma~\ref{prop0}).
So $\alpha_{a_i} \ge 0$ and $\alpha_{b_i} \ge 0$.
We also have $\alpha_{a_i} + \alpha_{b_i} = \wt_N(a_i,b_i) = 0$. 
Thus $\alpha_{a_i} = \alpha_{b_i} = 0$. 
\item In the second  possibility of Lemma~\ref{lem:redn3}, i.e., when $(a_i,b'_i), (a'_i,b_i)$ are in 
$N$, we have $\alpha_u = \pm 1$ for all $u \in C_i$. 
This is because $\wt_N(a_i,b_i) = 2$, thus $\alpha_{a_i} = \alpha_{b_i} = 1$.
\end{itemize}

\noindent{\bf Edge weights.}
We now assign weights to the edges in $G$. 
For all $e \in E_H$, let $w(s_e,t'_e) = w(s_e,t''_e) = w(s'_e,t_e) = w(s''_e,t_e) = 2$, i.e.,
all the edges between $s_e, t_e$ and their neighbors in $D_e$ have weight 2.
For all $i \in [n_H]$, let $w(a_i,b_i) = w(a'_i,b'_i) = 2$.
Set $w(e) = 1$ for all other edges $e$ in $G$.

Let $N$ be a max-weight popular matching in $G$ and let $\vec{\alpha}\in\{0,\pm 1\}^n$ be a witness of $N$'s popularity. Let
$U_N$ be the set of those indices $i \subseteq [n_H]$ such that $\alpha_u = \pm 1$ for all $u \in C_i$.

\begin{lemma}
  \label{lem:redn5}
      The set $U_N$ is a vertex cover of the graph $H$.
\end{lemma}
\begin{proof}
   Edge weights were assigned in $G$ such that the following claim (proof given below) holds for any max-weight popular matching $N$.
    \begin{new-claim}
        \label{clm:redn4}
        For every $e \in E_H$, both $s_e$ and $t_e$ have to be matched in $N$.
    \end{new-claim}

Consider any edge $e = (i,j) \in E_H$, let $i < j$. It follows from Claim~\ref{clm:redn4} and Lemma~\ref{lem:redn3} that either 
the pair $(s_e,t'_e), (s'_e,t_e)$ or  the pair $(s_e,t''_e), (s''_e,t_e)$ is in $N$. 
\begin{itemize}
  \item If  $(s_e,t'_e)$ and $ (s'_e,t_e)$ are in $N$, then $\wt_N(a_i,t_e) = 0$ since $t_e$ prefers $a_i$ to $s'_e$ while $a_i$ prefers both $b_i$
      and $b'_i$ (its possible partners in $N$) to $t_e$. It follows from part~2 of Lemma~\ref{prop0} that
      $\alpha_{t_e} = -1$, thus  $\alpha_{a_i}$ has to be 1 so that $\alpha_{a_i} + \alpha_{t_e} \ge 0$.
      Recall that $\vec{\alpha}$ is a witness of $N$'s popularity.
  \item If  $(s_e,t''_e)$ and $ (s''_e,t_e)$ are in $N$, then $\wt_N(s_e,b_j) = 0$ since $s_e$ prefers $b_j$ to $t''_e$ while $b_j$ prefers both $a_j$
        and $a'_j$ (its possible partners in $N$) to $s_e$. It follows from part~2 of Lemma~\ref{prop0} that
        $\alpha_{s_e} = -1$, thus  $\alpha_{b_j}$ has to be 1 so that $\alpha_{s_e} + \alpha_{b_j} \ge 0$.
    \end{itemize}

    Thus at least one of $C_i, C_j$ assigns the $\alpha$-values of its vertices to $\pm 1$. Hence for every edge $(i,j) \in E_H$, 
at least one of $i,j$ is in $U_N$, in other words, $U_N$ is a vertex cover of $H$. \qed
\end{proof}

\noindent{\bf Proof of Claim~\ref{clm:redn4}.}
Consider any $e \in E_H$. Since $s'_e,t'_e,s''_e$, and $t''_e$
are stable vertices in $G$, they have to be matched in every popular matching in $G$. 
Thus any popular matching in $G$ either matches {\em both} $s_e, t_e$ or neither $s_e$ nor $t_e$.
Recall from the proof of Lemma~\ref{lem:redn3} that there is no popular edge between $D_e$ and either $C_i$ or $C_j$.

  Let $N$ be a max-weight popular matching in $G$. We now need to show that $N$ matches  
both $s_e$ and $t_e$. Suppose $N$ matches neither $s_e$ nor $t_e$. Thus $(s'_e,t'_e)$ and 
$(s''_e,t''_e)$ are in $N$. Hence the weight contributed by vertices in $D_e$ to $w(N)$ is 2. 

Let $e = (i,j)$ where $i < j$. We know from Lemma~\ref{lem:redn3} 
that either $(a_i,b'_i),(a'_i,b_i)$ are in $N$ or $(a_i,b_i),(a'_i,b'_i)$ are in $N$.

  \smallskip
  
  \noindent{\em Case~1:} Suppose $\{(a_i,b'_i),(a'_i,b_i)\} \subset N$.
This means $\wt_N(a_i,b_i) = 2$ and hence $\alpha_{a_i} = \alpha_{b_i} = 1$ in any witness $\vec{\alpha}$ 
of $N$.

Consider the matching $N' = N\cup\{(s_e,t'_e),(s'_e,t_e)\} \setminus \{(s'_e,t'_e)\}$. That is, 
we replace the edge $(s'_e,t'_e)$ in $N$ with the edges $(s_e,t'_e)$ and $(s'_e,t_e)$. It is easy 
to prove that $N'$ is also popular; we will show a witness $\vec{\beta}$ to prove the popularity 
of $N'$. Let $\vec{\alpha}$ be a witness of $N$ and let $\beta_u = \alpha_u$ for all 
$u \in A \cup B$ except for the vertices in $D_e$. For the vertices in $D_e$, let 
$\beta_{s_e} =  \beta_{t_e} = \beta_{t''_e} = -1$ and $\beta_{s'_e} =  \beta_{t'_e} = \beta_{s''_e} = 1$.

It can be checked that $\sum_{u\in A\cup B}\beta_u = 0$ and $\beta_u +\beta_v \ge \wt_{N'}(u,v)$ for 
all $(u,v) \in E$ and $\beta_u \ge \wt_{N'}(u,u)$ for all $u \in A \cup B$.
In particular, $\wt_{N'}(a_i,t_e) = 0 = \beta_{a_i} + \beta_{t_e}$. Moreover, we have 
$w(N') = w(N) + 3$. Thus there is a popular matching $N'$ in $G$ with a larger weight than $N$, a 
contradiction to our assumption that $N$ is a max-weight popular matching in $G$.

  \smallskip
  
\noindent{\em Case~2:} Suppose $\{(a_i,b_i),(a'_i,b'_i)\} \subset N$.
Consider the matching $N''$ which is exactly the same as $N$ except
that the edges  $(s'_e,t'_e), (a_i,b_i)$ and $(a'_i,b'_i)$ are deleted from $N$ and the new edges 
are $(s_e,t'_e), (s'_e,t_e), (a_i,b'_i)$, and $(a'_i,b_i)$.

We claim that $N''$ is also popular and we prove this by showing a witness $\vec{\gamma}$.
Let $\gamma_u = \alpha_u$ for all $u \in A \cup B$ ($\vec{\alpha}$ is a witness of $N$) except 
for the vertices in $C_i \cup D_e$. For the vertices in $C_i \cup D_e$, let 
$\gamma_{s_e} =  \gamma_{t_e} = \gamma_{t''_e} = -1$, $\gamma_{s'_e} =  \gamma_{t'_e} = \gamma_{s''_e} = 1$,
and $\gamma_{a_i} =  \gamma_{b_i} = 1$,   $\gamma_{a'_i} =  \gamma_{b'_i} = -1$.

It can be checked that $\sum_{u\in A\cup B}\gamma_u = 0$ and $\gamma_u +\gamma_v \ge \wt_{N''}(u,v)$ 
for all $(u,v) \in E$ and $\gamma_u \ge \wt_{N''}(u,u)$ for all $u \in A \cup B$.
Thus $N''$ is a popular matching. Moreover, we have $w(N'') = w(N) + 3 -2 = w(N)+1$.
Thus there is a popular matching $N''$ in $G$ with a larger weight than $N$, a contradiction again.
Hence we can conclude that both $s_e$ and $t_e$ have to be matched in $N$. \qed

\begin{theorem}
    \label{thm:sec3}
    For any integer $1 \le k \le n_H$, the graph $H = (V_H,E_H)$ admits a vertex cover of size $k$ if and only if $G$ admits a popular
    matching of weight at least $5m_H + 4n_H - 2k$, where $|E_H| = m_H$.
\end{theorem}
\begin{proof}
  Let $U$ be a vertex cover of size $k$ in $H$. Using $U$, we will construct a matching $M$ in $G$ of weight $5m_H + 4n_H - 2k$
and a witness  $\vec{\alpha}$ to $M$'s popularity as follows. For every $i \in [n_H]$:
  \begin{itemize}
  \item if $i \in U$ then include edges $(a_i,b'_i)$ and $(a'_i,b_i)$ in $M$; set
    $\alpha_{a_i} = \alpha_{b_i} = 1$ and set $\alpha_{a'_i} = \alpha_{b'_i} = -1$.
  \item else include edges $(a_i,b_i)$ and $(a'_i,b'_i)$ in $M$ and set $\alpha_u = 0$ for all 
$u \in C_i$.
  \end{itemize}

  For every $e = (i,j) \in E_H$, where $i < j$, we do as follows:
  \begin{itemize}
  \item if $i \in U$ then include the edges $(s_e,t'_e),(s'_e,t_e)$, and $(s''_e,t''_e)$ in $M$; 
set $\alpha_{s_e} = \alpha_{t_e} = \alpha_{t''_e} = -1$
and $\alpha_{s'_e} = \alpha_{t'_e} = \alpha_{s''_e} = 1$.
  \item else (so $j \in U$) include the edges $(s_e,t''_e),(s'_e,t'_e)$, and $(s''_e,t_e)$ in $M$; 
set $\alpha_{s_e} = \alpha_{t_e} = \alpha_{s'_e} = -1$
and $\alpha_{t'_e} = \alpha_{s''_e} = \alpha_{t''_e} = 1$.
  \end{itemize}

  Set $\alpha_{a_0} = \alpha_{b_0} = 0$.
  It can be checked that $\alpha_u + \alpha_v = 0$ for every $(u,v) \in M$, hence
  $\sum_{u \in A\cup B}\alpha_u = 0$. Also $\alpha_u + \alpha_v \ge \wt_M(u,v)$
  for every edge $(u,v)$ in $G$ and $\alpha_u \ge \wt_M(u,u)$ for all $u \in A \cup B$. 
  Thus  $\vec{\alpha}$ is a witness to $M$'s popularity.
  
  We will now calculate $w(M)$. The sum of edge weights in $M$ from vertices in $D_e$ is 5, so this adds up to $5m_H$
  for all $e \in E_H$. For $i \in U$, the sum of edge weights in $M$ from vertices in $C_i$ is 2, so this adds
  up to $2k$ over all $i \in U_H$. For $j \notin U$, the sum of edge weights in $M$ from vertices in $C_j$ is 4, so this adds
  adds up to $4(n_H - k)$ over all $j \notin U_H$. Thus $f(M) = 5m_H + 4n_H - 2k$.

  \smallskip

  We will now show the converse. Suppose $G$ has a popular matching of weight at least 
  $5m_H + 4n_H - 2k$. Let $N$ be a max-weight popular matching in $G$.
  So $w(N) \ge 5m_H + 4n_H - 2k$.
  Let $\vec{\alpha}$ be a witness of $N$'s popularity and let $U_N = \{i \subseteq[n_H]: \alpha_u = \pm 1\, \forall u \in C_i\}$.
  We know from Lemma~\ref{lem:redn5} that $U_N$ is a vertex cover of~$H$. We will now show that 
  $|U_N| \le k$. 
  
  We know from Claim~\ref{clm:redn4} that for every $e \in E_H$, both $s_e$ and $t_e$ have to be 
  matched in~$N$. So each gadget $D_e$ contributes a weight
  of 5 towards $w(N)$. Each gadget $C_i$, for $i \in U_N$, contributes a weight of 2 towards 
  $w(N)$ while each gadget $C_j$, for $j \notin U_N$, contributes a weight of 4 towards $w(N)$.
  Hence $w(N) = 5m_H + 2|U_N| + 4(n_H - |U_N|)$. Since $5m_H + 2|U_N| + 4(n_H - |U_N|) \ge 5m_H + 4n_H - 2k$,
  we get $|U_N| \le k$. Thus $H$ has a vertex cover of size $k$. \qed
\end{proof}    

Theorem~\ref{thm:hard} stated in Section~\ref{intro} now follows.
It is easy to see that Theorem~\ref{thm:sec3} holds even if ``popular matching of weight at least 
$5m_H + 4n_H - 2k$'' is replaced by ``{\em max-size} popular matching of weight at least 
$5m_H + 4n_H - 2k$''. Thus the problem of computing a max-size
popular matching problem in $G = (A \cup B,E)$ of largest weight is also $\mathsf{NP}$-hard.

\section{Max-weight popular matchings: exact and approximate solutions}
\label{section6}
Let $G = (A \cup B, E)$ be an instance with a weight function $w: E \rightarrow \mathbb{R}_{\ge 0}$.
We will now show that a popular matching in $G$ of weight at least $\opt/2$  can be computed in polynomial time, where
$\opt = w(M^*)$ and $M^*$ is a max-weight popular matching in $G$.
The following decomposition theorem for any popular matching $M$ in $G = (A \cup B, E)$ is known.

\begin{theorem}[\cite{CK16}]
$M$ can be partitioned into $M_0 \cup M_1$ such that $M_0 \subseteq S$ and $M_1 \subseteq D$, 
where $S$ is a stable matching and $D$ is a (strongly) dominant matching in $G$.
\end{theorem}

Any popular matching $M$ has a witness $\vec{\alpha} \in \{0,\pm 1\}^n$: let $M_0$ (similarly, $M_1$) be the set of edges of $M$ 
on vertices $u$ with $\alpha_u = 0$ (resp., $\alpha_u = \pm 1$). Since $\alpha_a + \alpha_b = 0$ for all $(a,b)\in M$, the matching
$M$ is a disjoint union of $M_0$ and $M_1$. 
If $M$ is a max-weight popular matching in $G$, then one of $M_0,M_1$ has weight at least $w(M)/2$.

What the above result from \cite{CK16} shows is that $M_0$ can be extended to a stable matching in $G$ and $M_1$ can be extended
to a dominant matching in $G$. Consider the following algorithm.

\begin{enumerate}
\item Compute a max-weight stable matching $S^*$ in $G$.
\item Compute a max-weight dominant matching $D^*$ in $G$.
\item Return the matching in $\{S^*, D^*\}$ with larger weight.
\end{enumerate}

Since all edge weights are non-negative, either the max-weight stable matching in $G$ or the max-weight dominant matching in $G$
has weight at least $w(M^*)/2 = \opt/2$. 
Thus Steps~1-3 compute a 2-approximation for max-weight popular matching in $G = (A \cup B,E)$.

Regarding the implementation of this algorithm, both $S^*$ and $D^*$ can be computed in polynomial time~\cite{Rot92,CK16}.
We also show descriptions of the stable matching polytope and the dominant matching polytope below.
Thus the above algorithm runs in polynomial time.

\subsection{A fast exponential time algorithm}
We will now show an algorithm to compute a max-weight popular matching in $G = (A \cup B,E)$ with $w: E \rightarrow \mathbb{R}$.
We will use an extended formulation from \cite{KMN09} of the  popular fractional matching polytope ${\cal P}_G$. This is obtained 
by generalizing LP2 from Section~\ref{prelims} to all fractional matchings $\vec{x}$ in $G$, i.e.,  $\sum_{e \in \tilde{E}(u)}x_e = 1$ 
for all vertices $u$ and $x_e \ge 0$ for all $e \in \tilde{E}$. So $\wt_x(a,b)$, which is the sum of votes of $a$ and $b$ for each other 
over their respective assignments in $\vec{x}$, replaces $\wt_M(a,b)$.
Thus $\wt_x(a,b) = \sum_{b':b'\prec_a b}x_{ab'} - \sum_{b':b'\succ_a b}x_{ab'} + \sum_{a':a'\prec_b a}x_{a'b} - \sum_{a':a'\succ_b a}x_{a'b}$.

\smallskip

Recall the subgraph $F_G = (A \cup B, E_F)$ where $E_F$ is the set of popular edges in $G$. The set $E_F$ can be efficiently computed 
as follows: call an edge {\em stable} if it belongs to some stable matching in $G$. 
It was shown in \cite{CK16} that  every popular 
edge in $G = (A \cup B, E)$ is a stable edge either in $G$ or in a larger bipartite graph $G'$ on $O(n)$ vertices and $O(m+n)$ edges. 
All stable edges in $G$ and in $G'$ can be identified in linear time~\cite{GI89}.

Let ${\cal C} = \{\Comp_1,\ldots,\Comp_h\}$ be the set of connected components in $F_G$ and
let $\ell$ be the number of non-singleton sets in ${\cal C}$. 
Let us assume that components $\Comp_1,\ldots,\Comp_{\ell}$ have size at least 2.
For any popular matching $M$ with a witness $\vec{\alpha} \in\{0,\pm 1\}^n$, Lemma~\ref{prop1} 
tells us that $\vec{\alpha}$ defines a natural function (we will call this function $\alpha$) from 
${\cal C}$ to $\{0,1\}$, where $\alpha(\Comp_i) = 0$ if $\alpha_u = 0$
for all $u \in \Comp_i$ and $\alpha(\Comp_i) = 1$ if $\alpha_u = \pm 1$ for all $u \in \Comp_i$.
Note that $\alpha(\Comp_i) = 0$ for $\ell+1 \le i \le h$.

\begin{definition}
  \label{def:polytope-Pw}
For $\vec{r} = (r_1,\ldots,r_{\ell}) \in \{0,1\}^{\ell}$, let ${\cal P}(\vec{r})$ be the convex hull  of all popular 
matchings $M$ in $G$ such that $M$ has a witness $\vec{\alpha}$ with $\alpha(\Comp_i) = r_i$ for all $1 \le i \le \ell$. 
\end{definition}

${\cal P}(\vec{0})$ is the stable matching polytope and ${\cal P}(\vec{1})$ is the dominant matching polytope.
We will now show an extended formulation of ${\cal P}(\vec{r})$ for any $\vec{r} \in  \{0,1\}^{\ell}$. First augment $\vec{r}$
with $r_{\ell+1} = \cdots = r_h = 0$ so that $\vec{r} \in \{0,1\}^h$.

\smallskip

\noindent{\bf Our constraints.} 
Let $U$ be the set of all unstable vertices in $\cup_i\Comp_i$ where $i \in \{1,\ldots,h\}$ is such that $r_i = 0$. 
Let $A' = A \setminus U$ and $B' = B \setminus U$. 
It follows from Lemma~\ref{prop1} that all vertices in $U$ are left unmatched in any popular 
matching $M$ in ${\cal P}(\vec{r})$ while all vertices in $A' \cup B'$ are matched in $M$ (recall
that all stable vertices of $G$ are matched in $M$).
Consider constraints~(\ref{new-constr1})-(\ref{new-constr4}) given below.
\begin{eqnarray}
  \alpha_a + \alpha_b \ & \ge & \ \wt_x(a,b) + |r_i - r_j| \ \ \ \ \forall (a,b) \in E \cap (\Comp_i \times \Comp_j), \ 1 \le i, j \le h \label{new-constr1}\\
  \alpha_a + \alpha_b \ & = & \ \wt_x(a,b) \ \ \ \  \ \ \ \  \ \ \ \ \ \ \ \ \ \ \forall (a,b) \in E_F\label{new-constr2}\\ 
  \sum_{u \in A\cup B}\alpha_u & = & \ 0 \ \ \ \ \ \ \ \ \ \ \ \ \ \ \ \ \ \text{and} \ \ \ \ \ -r_i \, \le \, \alpha_u \, \le \, r_i \ \ \ \forall u \in \Comp_i,\ 1 \le i \le h \label{new-constr3}\\
   x_{(u,u)} & = & 1 \ \ \ \forall u \in U \ \ \ \ \ \,\text{and} \ \ \ \ \  x_{(u,u)} \ = \ 0 \ \ \ \forall u \in A'\cup B'  \label{newer-constr3}\\
  \sum_{e \in \tilde{E}(u)}x_e & = & \ 1  \ \ \forall u \in A\cup B,  \ \ \ \ x_e \ge 0 \ \ \forall e \in E_F, \ \ \ \text{and}\ \ \ x_e = 0 \ \ \forall e \in E\setminus E_F\label{new-constr4}
\end{eqnarray}

The variables in the constraints above are $x_e$ for $e \in E$ and $\alpha_u$ for $u \in A \cup B$. 
Constraint~(\ref{new-constr1}) tightens the edge covering constraint $\alpha_a + \alpha_b \ge \wt_x(a,b)$ for edges in 
$\Comp_i \times \Comp_j$ with $r_i \ne r_j$. Consider any popular matching $M$ with witness $\vec{\alpha}$ such that
$\alpha(\Comp_i) \ne \alpha(\Comp_j)$. So $M$ and $\vec{\alpha}$ satisfy
$\alpha_a + \alpha_b \ge \wt_M(a,b)$. Since $\alpha_a + \alpha_b \in \{\pm 1\}$ while $\wt_M(a,b)\in \{\pm 2,0\}$, 
$M$ and  $\vec{\alpha}$ have to satisfy $\alpha_a + \alpha_b \ge \wt_M(a,b) + 1 =  \wt_M(a,b) + |r_i - r_j|$.

Constraint~(\ref{new-constr2}) makes the edge covering constraint {\em tight} for all popular edges $(a,b)$.
This is because for any popular matching $M$ and witness $\vec{\alpha}$, we have $\alpha_a + \alpha_b = \wt_M(a,b)$ for any
popular edge $(a,b)$ (see the proof of Lemma~\ref{prop0}). 

Constraint~(\ref{new-constr4}) is clearly satisfied by any
popular matching $M$ and any witness $\vec{\alpha}$ satisfies $\sum_{u \in A\cup B}\alpha_u = 0$.
The other constraints in (\ref{new-constr3}) and (\ref{newer-constr3}) are consequences of the parity $r_i$ of the 
component $\Comp_i$ that a vertex belongs to. We will prove the following theorem in Section~\ref{appendix-c}.

\begin{theorem}
\label{thm:polytope}
Constraints~(\ref{new-constr1})-(\ref{new-constr4}) define an extended formulation of the polytope ${\cal P}(\vec{r})$.
\end{theorem}

Thus a max-weight popular matching in  ${\cal P}(\vec{r})$ can be computed in polynomial time and hence a max-weight popular matching
in $G$ can be computed in $O^*(2^{\ell})$ time by going through all $\vec{r} \in \{0,1\}^{\ell}$. Recall that $\ell$ 
is the number of components in $F_G$ of size at least 2. Since $\ell \le n/2$, this is an $O^*(2^{n/2})$ algorithm for computing
a max-weight popular matching.

\medskip

\noindent{\bf A faster exponential time algorithm.}
We will show in Section~\ref{appendix-d} that it is enough to go through all $\vec{r} \in \{0,1\}^k$, where $k$ is the number of 
components in $F_G$ of size at least 4. We do this by introducing a new variable $p_i$, where $0 \le p_i \le 1$, to replace $r_i$ 
and represent the parity of $\Comp_i$, for each $\Comp_i \in {\cal C}$ of size~2. 

We will show that the resulting polytope 
is an extended formulation of
the convex hull of all popular matchings $M$ in $G$ such that $M$ has a witness $\vec{\alpha}$ with 
$\alpha(\Comp_i) = r_i$ for $1 \le i \le k$. This yields an $O^*(2^k)$ algorithm for a max-weight popular matching 
in $G=(A\cup B,E)$. Since $k \le n/4$, this proves Theorem~\ref{thm:exp-alg} stated in Section~\ref{intro}.
Also, when $k = O(\log n)$, we have a polynomial time algorithm to compute a max-weight popular matching in $G$.

\subsection{Proof of Theorem~\ref{thm:polytope}}
\label{appendix-c}
Let ${\cal Q}'(\vec{r}) \subseteq \mathbb{R}^{m+n}$ be the polytope defined by constraints~(\ref{new-constr1})-(\ref{new-constr4}),
where $|E| = m$.
Let ${\cal Q}(\vec{r})$ denote the polytope ${\cal Q}'(\vec{r})$ projected on to its first 
$m$ coordinates (those corresponding to $e \in E$).
Our goal is to show that ${\cal Q}(\vec{r})$ is the polytope ${\cal P}(\vec{r})$. 

We will first show that every popular matching $M$ with at least one witness $\vec{\alpha} \in \{0,\pm 1\}^n$ such that 
$\alpha(\Comp_i) = r_i$ for $i = 1,\ldots,\ell$ belongs to ${\cal Q}(\vec{r})$. It is easy to see that $M$ and $\vec{\alpha}$ satisfy
constraints~(\ref{new-constr2})-(\ref{new-constr4}).
Regarding constraint~(\ref{new-constr1}), for any edge $(a,b) \in \Comp_i\times\Comp_j$ such that
$r_i \ne r_j$,  we have $\alpha_a + \alpha_b \in \{\pm 1\}$ while $\wt_M(a,b)\in \{\pm 2,0\}$.
Thus $\alpha_a + \alpha_b \ge \wt_M(a,b)$ implies that $\alpha_a + \alpha_b \ge \wt_M(a,b) + 1$.
Hence $(M, \vec{\alpha}) \in {\cal Q}'(\vec{r})$, i.e., $M \in {\cal Q}(\vec{r})$. Thus ${\cal P}(\vec{r}) \subseteq {\cal Q}(\vec{r})$.

\smallskip

Let $(\vec{x},\vec{\alpha}) \in {\cal Q}'(\vec{r})$. We will now use the techniques and results from \cite{TS98,HK17} to show that
$\vec{x}$ is a convex combination of some popular matchings in ${\cal P}(\vec{r})$. 
This will prove ${\cal Q}(\vec{r}) = {\cal P}(\vec{r})$. 

The polytope ${\cal P}(\vec{0})$ is the stable matching polytope and a simple proof of integrality of Rothblum's 
formulation~\cite{Rot92} of this polytope was given in~\cite{TS98}. 
When $G$ admits a perfect stable matching, the polytope ${\cal P}(\vec{1})$ is the same as ${\cal P}_G$ and a proof of integrality 
of ${\cal P}_G$ in this case was given in \cite{HK17}.
Note that $r_1,\ldots,r_h$ are all 0  in \cite{TS98} and $r_1,\ldots,r_h$ are all 1 for this result in \cite{HK17}, i.e., the
integrality of ${\cal P}_G$ when $G$ admits a perfect stable matching.

We will build a table $T$ of width~1 (as done in \cite{TS98,HK17}) with $n' = |A'\cup B'|$ rows: one corresponding to 
each vertex $u \in A' \cup B'$. The row corresponding to $u$ will be called $T_u$. 

Form an array $X_u$ of length 1 as follows: if $x_{(u,v)} > 0$ then there is a cell of length $x_{(u,v)}$ in $X_u$ with entry $v$ 
in it. The entries in $X_u$ will be sorted in {\em increasing order} of preference for $u \in A'$ and in {\em decreasing order}  
of preference for $u \in B'$. This was the order used in \cite{TS98}. If $u \in \Comp_i$ where 
$i \in \{1,\ldots,\ell\}$ with $r_i = 1$, then reorder $X_u$ as done in \cite{HK17}. 

For any $a \in A'$ that belongs to such a component $\Comp_i$,
the initial or least preferred $(1+ \alpha_a)/2$ fraction of $X_a$ will be called the {\em positive} or {\em blue}
sub-array of $X_a$ and the remaining part, which is the most preferred $(1-\alpha_a)/2$ fraction of $X_a$,
will be called the {\em negative} or {\em red} sub-array of $X_a$. The array $X_a$ will be reordered as shown in 
Fig.~\ref{old-fig:example}, i.e., the positive and negative sub-arrays of $X_a$ are swapped.
Call the reordered array $T_a$.

\begin{figure}[h]
\centerline{\resizebox{0.72\textwidth}{!}{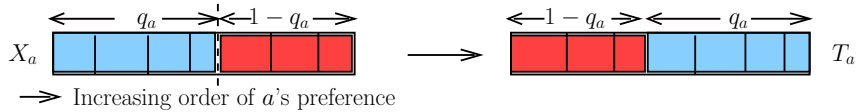}}
\caption{The array $X_a$ (on the left) will be reordered by swapping the positive and negative sub-arrays as shown above. The value 
$q_a = (1+\alpha_a)/2$, so $1 - q_a = (1-\alpha_a)/2$.}
\label{old-fig:example}
\end{figure}

A similar transformation from $X_b$ to $T_b$ was shown in \cite{HK17} for each $b \in B'$ that 
belongs to a component $\Comp_i$ where $i \in \{1,\ldots,\ell\}$ and $r_i = 1$.
The initial or most preferred $(1 - \alpha_b)/2$ fraction of $X_b$ will be called the {\em negative}
sub-array of $X_b$ and the remaining part, which is the 
least preferred $(1+\alpha_b)/2$ fraction of $X_b$, will be called the {\em positive} sub-array of $X_b$. 
As before, swap the positive and negative sub-arrays of $X_b$ and call this reordered array $T_b$.

If $u \in A'\cup B'$ is in a component $\Comp_i$ with $r_i=0$, then we leave $X_u$ 
as it is. That is, we set $T_u = X_u$   (see Fig.~\ref{fig:example}).
There are no positive or negative sub-arrays in $T_u$.

\vspace*{-0.2cm}

\begin{figure}[h]
\centerline{\resizebox{0.3\textwidth}{!}{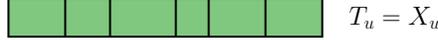}}
  \caption{The entries in $T_u$ are sorted in increasing order of preference for $u \in A$ and in
    decreasing order of preference for $u \in B$.}
\label{fig:example}
\end{figure}

\noindent{\bf Finding the popular matchings that $\vec{x}$ is a convex combination of.}
Let $T$ be the table with rows $T_u$, for $u \in A' \cup B'$.
For any $t \in [0,1)$, define the matching $M_t$ as follows:
\begin{itemize}
\item let $L(t)$ be the vertical line at distance $t$ from the left boundary of $T$;
\item $L(t)$ intersects (or touches the left boundary of) some cell in $T_u$, call this cell $T_u[t]$, for each $u \in A' \cup B'$;
  \begin{equation}
    \label{last-eqn}
    \text{define}\ M_t = \{(u,v): u\in A' \cup B'\ \text{and}\ v\ \text{is\ in\ cell}\ T_u[t]\}.
  \end{equation}  
\end{itemize}

\noindent{\bf Validity of $M_t$.}
We need to prove that $M_t$ is a valid matching in $G$. That is, for any vertex $u \in A'\cup B'$, we need to show that if 
$v$ belongs to cell $T_u[t]$, then $u$ belongs to cell $T_v[t]$. Note that both $u$ and $v$ have to belong to the
same component in the subgraph $F_G$ since $x_{uv} = 0$ otherwise (since $x_e = 0$ for $e \in E \setminus E_F$).
Let $\Comp_i$ be the connected component in $F$ containing $u$ and $v$. There are 2 cases here: (i)~$r_i = 1$ and (ii)~$r_i = 0$.

The proof in case~(i) follows directly from Theorem~3.2 in \cite{HK17}.
The proof in case~(ii) is given in the proof of Theorem~1 in \cite{TS98}. Both these proofs are based 
on the ``tightness'' of the edge $(u,v)$, i.e., $\alpha_u + \alpha_v = \wt_x(u,v)$ in case~(i) and 
$\wt_x(u,v) = 0$ in case~(ii). The tightness of $(u,v)$ holds for us as well: 
since $x_{uv} > 0$, $(u,v) \in E_F$, 
so $\alpha_u + \alpha_v = \wt_x(u,v)$ by constraint~(\ref{new-constr2}) in
case~(i); and in case~(ii), we have $\alpha_u = \alpha_v = 0$ by constraint~(\ref{new-constr3})
and so $\wt_x(u,v) = 0$ by constraint~(\ref{new-constr2}).

\medskip

\noindent{\bf Popularity of $M_t$.}
We will now show that $M_t$ is popular.
Define a vector $\vec{\alpha}^t \in \{0,\pm 1\}^n$:
\begin{itemize}
\item For $1 \le i \le h$: if $r_i = 0$ then set $\alpha^t_u = 0$ for each $u \in \Comp_i$;
else for each $u \in \Comp_i$
\smallskip
\begin{itemize}
\item[$\ast$] if the cell $T_u[t]$ is {\em positive} (or blue) then set $\alpha^t_u = 1$,
else set $\alpha^t_u = -1$. 
\end{itemize}
\end{itemize}

We will now show that $\vec{\alpha}^t$ is a witness of $M_t$.
Thus $M_t$ will be a popular matching in $G$, in fact, $M_t$ will be in ${\cal P}(\vec{r})$. 
This is because by our assignment of $\alpha^t$-values above, $\alpha^t(\Comp_i) = r_i$ for 
$i \in \{1,\ldots,h\}$. Our first claim is that $\sum_{u \in A \cup B}\alpha^t_u = 0$.

To prove the above claim, observe that for every vertex $u$ that is left unmatched in $M_t$ (all these vertices 
belong to $U$), we have $\alpha^t_u = 0$.  We now show that for every $(a,b)$ in $M_t$, we have $\alpha^t_a + \alpha^t_b = 0$. Since 
$(a,b) \in E_F$, the vertices $a$ and $b$ belong to the same component $\Comp_i$. 
If $r_i = 0$ then $\alpha^t_a = \alpha^t_b = 0$ and hence $\alpha^t_a + \alpha^t_b = 0$. 
If $r_i = 1$ then $\alpha^t_a, \alpha^t_b \in \{\pm 1\}$ and the proof
that $\alpha^t_a + \alpha^t_b = 0$ was shown in \cite{HK17} (see Corollary~3.1). 
Thus $\sum_{u \in A \cup B}\alpha^t_u = 0$.

We will now show that $\alpha^t_u \ge \wt_{M_t}(u,u)$ for all $u \in A\cup B$. Every vertex $u \in A' \cup B'$ 
is matched in $M_t$ and so $\wt_{M_t}(u,u) = -1$. Since $\alpha^t_u \ge -1$ for all 
$u \in A'\cup B'$, we have $\alpha^t_u \ge \wt_{M_t}(u,u)$ for these vertices.
For $u \in U$, we have $\alpha^t_u = 0 = \wt_{M_t}(u,u)$.

What is left to show is that $\alpha^t_a + \alpha^t_b \ge \wt_{M_t}(a,b)$ for all $(a,b) \in E$. 
Lemma~\ref{lemma:mat3} below shows this. Hence we can conclude that $M_t$ is a popular matching in $G$,
in particular, $M_t \in {\cal P}(\vec{r})$.

It is now easy to show that $\vec{x}$ is a convex combination of matchings that belong to $\mathcal{P}(\vec{r})$.
To obtain these matchings, as done in \cite{TS98}, sweep a vertical line from the left boundary
of table $T$ to its right boundary: whenever the line hits the left wall of one or more new cells,
a new matching is obtained. If the left wall of the $i$-th leftmost cell(s) in the table $T$ is at distance
$t_i$ from the left boundary of $T$, then we obtain the matching $M_{t_i}$ defined analogous to $M_t$ in (\ref{last-eqn}). 

Let $M_0, M_{t_{1}},\ldots,M_{t_{d-1}}$ be all the matchings obtained by sweeping a vertical line 
through the table $T$. This means that $\vec{x} = t_1\cdot M_0 + (t_2-t_1)\cdot M_{t_1} + \cdots + (1-t_{d-1})\cdot M_{t_{d-1}}$.
Thus $\vec{x}$ is a convex combination of  matchings in $\mathcal{P}(\vec{r})$. 
This finishes the proof that ${\cal Q}(\vec{r}) = {\cal P}(\vec{r})$.

\begin{lemma}
  \label{lemma:mat3}
  For any $(a,b) \in E$, we have $\alpha^t_a + \alpha^t_b \ge \wt_{M_t}(a,b)$.
\end{lemma}
\begin{proof}
Let $a \in \Comp_i$ and $b \in \Comp_j$, where $1 \le i,j \le h$.
The proof that $\alpha^t_a + \alpha^t_b \ge \wt_{M_t}(a,b)$ when $r_i = r_j$ follows from \cite{HK17,TS98}. 
When $r_i = r_j = 1$, it was shown in \cite{HK17} (see Lemma~3.5) that
$\alpha^t_a + \alpha^t_b \ge \wt_{M_t}(a,b)$.
When $r_i = r_j = 0$, it was shown in \cite{TS98} (see Theorem~1)
that $\wt_{M_t}(a,b) \le \alpha^t_a + \alpha^t_b = 0$. 
We need to show $\alpha^t_a + \alpha^t_b \ge \wt_{M_t}(a,b)$ when $r_i \ne r_j$. 
Assume without loss of generality that $r_i = 1$ and $r_j = 0$.
So $\alpha_b = 0$. 

The constraint corresponding to edge $(a,b)$ is 
$\alpha_a + \alpha_b \ge \wt_x(a,b) + |r_i-r_j|$ which simplifies to $\alpha_a \ge \wt_x(a,b) + 1$. 
If $b \in U$ then $\wt_x(a,b) \ge 0$ since $b$ prefers $a$ to itself. So $\alpha_a \ge \wt_x(a,b) + 1$ along with $\alpha_a \le 1$ 
(by constraint~(\ref{new-constr3})) implies that $\alpha_a = 1$ and $\wt_x(a,b) = 0$. This means that $\alpha^t_a = 1$ and 
$a$ prefers all the entries in the array $T_a$ to $b$. 
Hence $\wt_{M_t}(a,b) = 0$ and thus we have $\alpha^t_a + \alpha^t_b = 1+0 \ge \wt_{M_t}(a,b)$.

Hence let us assume that $b \in B'$. Since $\sum_{e \in \tilde{E}(u)}x_e = 1$ for all $u$, $\wt_x(a,b)$ equals
\[\sum_{b' \prec_a b}x_{ab'} - \sum_{b' \succ_a b}x_{ab'} + \sum_{a' \prec_b a}x_{a'b} - \sum_{a' \succ_b a}x_{a'b} \ \ = \ \ 2\left(\sum_{b' \prec_a b}x_{ab'} + \sum_{a' \prec_b a}x_{a'b} + x_{ab} - 1\right).\] 
Here $x_{ab} = 0$ since $a$ and $b$ belong to distinct components in $F_G$. So the  constraint $\alpha_a \ge \wt_x(a,b) + 1$ 
simplifies to:
\[ 2q_a - 1 \ \ \ \ge \ \ \ 2\left(\sum_{b' \prec_a b}x_{ab'} + \sum_{a' \prec_b a}x_{a'b} -1\right) +1, \ \text{where}\ \alpha_a = 2q_a - 1.\]
This becomes $q_a \ge \sum_{b' \prec_a b}x_{ab'} + \sum_{a' \prec_b a}x_{a'b}$. 
See Fig.~\ref{new-fig:stable1}.
\begin{figure}[h]
\centerline{\resizebox{0.34\textwidth}{!}{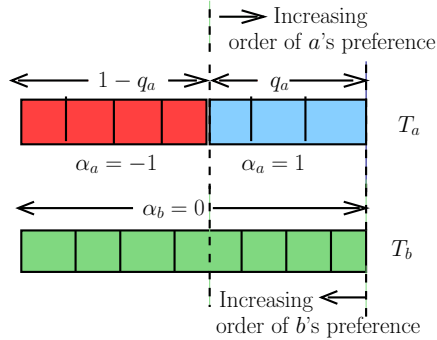}}
\caption{The cells where either $a$ or $b$ gets matched to a neighbor worse than the other is restricted to the blue sub-array of $T_a$ and the sub-array of $T_b$ exactly below this (between the two dashed vertical lines). The sum of lengths of such cells added up over both $T_a$ and $T_b$ is at most $q_a$.}
\label{new-fig:stable1}
\end{figure}

The neighbors with which $a$ gets paired in $\vec{x}$ start in increasing order of $a$'s 
preference from the dashed line separating the red sub-array and blue sub-array in $T_a$ 
and this wraps around in left to right orientation (see Fig.~\ref{new-fig:stable1}). 
For $b$, the neighbors with which $b$ gets paired in $\vec{x}$ start in increasing order 
of $b$'s preference from the right end of its array $T_b$ and this order is from {\em right to left}.
Thus $b$ is matched in $\vec{x}$ to its worst neighbor at the right end of $T_b$
and to its best neighbor at the left end of $T_b$.

Since $\sum_{b' \prec_a b}x_{ab'} + \sum_{a' \prec_b a}x_{a'b} \le q_a$, this implies that the subarray 
where either $a$ is matched to a worse neighbor than $b$ (this subarray has length 
$\sum_{b' \prec_a b}x_{ab'}$) or $b$ is matched to a worse neighbor than $a$ (this subarray has length 
$\sum_{a' \prec_b a}x_{a'b}$) is confined to within the dashed line separating the red and blue 
subarrays of $T_a$ and the rightmost wall of $T_b$ (see Fig.~\ref{new-fig:stable1}). 
Also, there is no cell where {\em both} $a$ and $b$ are matched to worse neighbors than each other 
as this would make $\sum_{b' \prec_a b}x_{ab'} + \sum_{a' \prec_b a}x_{a'b} > q_a$. 

If $T_a[t]$ is {\em positive} or blue, then $\alpha^t_a = 1$. We have $\wt_{M_t}(a,b) \le 0$ here
since when one of $a,b$ is getting matched to a worse neighbor than the other, the other 
is getting matched to a better neighbor.
Since $\alpha^t_b = 0$ in the entire array, we have $\alpha^t_a + \alpha^t_b  = 1 > \wt_{M_t}(a,b)$.

If $T_a[t]$ is {\em negative} or red, then $T_a[t]$ contains a neighbor that $a$ prefers to $b$ and
similarly, $T_b[t]$ contains a neighbor that $b$ prefers to $a$. That is, $\wt_{M_t}(a,b) = -2$ 
and here we have $\alpha^t_a = -1$ and $\alpha^t_b = 0$. Thus we have 
$\alpha^t_a + \alpha^t_b = -1 > \wt_{M_t}(a,b)$. This shows that the edge $(a,b)$ is always 
covered by the sum of $\alpha^t$-values of $a$ and $b$. \qed
\end{proof} 

\subsection{A faster exponential time algorithm}
\label{appendix-d}
Recall the popular subgraph $F_G = (A \cup B, E_F)$ and the set ${\cal C} = \{\Comp_1,\ldots,\Comp_h\}$ of connected components in $F_G$.
Assume without loss of generality that $\Comp_1,\ldots,\Comp_k$ have size at least 4 and
$\Comp_{k+1},\ldots,\Comp_{\ell}$ have size 2. Let us define the polytope ${\cal P}'(\vec{r})$ as follows.

\begin{definition}
  \label{def:new-polytope-Pw}
For $\vec{r} = (r_1,\ldots,r_k) \in \{0,1\}^k$, let ${\cal P}'(\vec{r})$ be the convex hull  of all popular 
matchings $M$ in $G$ such that $M$ has a witness $\vec{\alpha}$ with $\alpha(\Comp_i) = r_i$ for all $1 \le i \le k$. 
\end{definition}

Consider constraints~(\ref{new-constr1})-(\ref{new-constr4}) from Section~\ref{section6}. Regarding constraint~(\ref{new-constr1}),
$r_1,\ldots,r_k$ are the coordinates in the vector $\vec{r} \in \{0,1\}^k$ while $r_{\ell+1} = \cdots = r_{h} = 0$ as before, and 
$r_i = p_i$ for $k+1 \le i \le \ell$, where $p_{k+1},\ldots,p_{\ell}$ are variables. So in constraint~(\ref{new-constr1}):
\begin{itemize}
\item $|r_i - r_j| \in \{0,1\}$ for $i, j \in \{1,\ldots,k\} \cup \{\ell+1,\ldots,h\}$.
\item when one of $i,j$ (say, $i$) is in $\{k+1,\ldots,\ell\}$ and $j \in \{1,\ldots,k\}\cup\{\ell+1,\ldots,h\}$
  then $|r_i - r_j| = 1-p_i$ if $r_j = 1$ and $|r_i - r_j| = p_i$ if $r_j = 0$.
\item when both $i,j$ are in  $\{k+1,\ldots,\ell\}$,
we replace constraint~(\ref{new-constr1}) with two constraints: one where $|r_i - r_j|$ is 
replaced by $p_i-p_j$ and another where  $|r_i - r_j|$ is replaced by $p_j-p_i$. 
\end{itemize}

Similarly, in constraint~(\ref{new-constr3}), $-r_i \le \alpha_u \le r_i$ now becomes  $-p_i \le \alpha_u \le p_i$,
for $u \in \Comp_i$ where $i \in \{k+1,\ldots,\ell\}$.
Constraints~(\ref{new-constr2}), (\ref{newer-constr3}), 
and (\ref{new-constr4}) remain the same as before. 
Also, the sets $U, A', B'$ are the same as before since all in $\Comp_{k+1}\cup\cdots\cup\Comp_{\ell}$ are stable vertices. 
We will show the following theorem here.
The proof of Theorem~\ref{thm:new-polytope} will follow the same outline as the proof of Theorem~\ref{thm:polytope}.

\begin{theorem}
\label{thm:new-polytope}
The revised constraints~(\ref{new-constr1})-(\ref{new-constr4}) along with the constraints $0 \le p_i \le 1$ for $k+1 \le i \le \ell$ 
define an extended formulation of the polytope ${\cal P}'(\vec{r})$.
\end{theorem}
\begin{proof}
Let ${\cal S}'(\vec{r})$ be the polytope defined by the revised constraints~(\ref{new-constr1})-(\ref{new-constr4}) 
along with $0 \le p_i \le 1$ for $k+1 \le i \le \ell$.
Let ${\cal S}(\vec{r})$ denote the polytope ${\cal S}'(\vec{r})$ projected on to the coordinates
corresponding to $e \in E$.
We will now show that ${\cal S}(\vec{r})$ is the polytope ${\cal P}'(\vec{r})$. 

It is easy to see that every popular matching $M$ with at least one witness $\vec{\alpha}$ such that $\alpha(\Comp_i) = r_i$ 
for $i = 1,\ldots,\ell$ belongs to ${\cal S}(\vec{r})$. Thus ${\cal P}'(\vec{r}) \subseteq {\cal S}(\vec{r})$.
Let $(\vec{x},\vec{\alpha},\vec{p}) \in {\cal S}'(\vec{r})$. We will now show that
$\vec{x}$ is a convex combination of matchings in ${\cal P'}(\vec{r})$.

As done in the proof of Theorem~\ref{thm:polytope}, we will construct a table $T$ with $n' = |A' \cup B'|$ rows. The rows of vertices
outside $\Comp_{k+1}\cup\cdots\cup\Comp_{\ell}$ will be the same as before. Recall that $|\Comp_i| = 2$ for $k+1 \le i \le \ell$ and 
for any $i \in \{k+1,\ldots,\ell\}$: every vertex $u \in \Comp_i$ is matched to its only neighbor $v \in \Comp_i$ with $x_{uv} = 1$. 
So we do not reorder $X_u$ as it consists of just a single entry, i.e., $T_u = X_u$. However we partition $T_u$ into 
3 cells: (i)~the {\em positive} or blue cell, (ii)~the {\em negative} or red cell, and (iii)~the {\em zero} or green cell.

\begin{figure}[h]
\centerline{\resizebox{0.34\textwidth}{!}{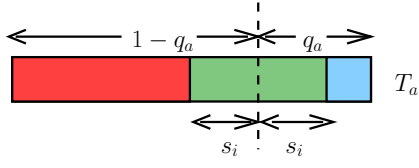}}
\caption{The red part (of length $1-q_a - s_i$) is the negative  cell of $T_a$ and the blue part (of length $q_a - s_i$)
  is the positive cell of $T_a$. The middle or green part (of length $2s_i$) is its zero cell.} 
\label{new-Cj-1:example}
\end{figure}

\begin{itemize}
\item for $a \in \Comp_i$, compute $2s_i = 1-p_i$. Since $p_i \le 1$, we have $s_i \ge 0$. Also $p_i \ge \pm \alpha_u$ since $\alpha_u$ is
sandwiched between $-p_i$ and $p_i$. So $1 - 2s_i \ge \pm \alpha_a$, i.e., $s_i \le q_a$ and $s_i \le 1 - q_a$ where $q_a = (1+\alpha_a)/2$.
The array $T_a$ gets divided into three cells as shown in Fig.~\ref{new-Cj-1:example}. 

\smallskip

The leftmost $1-q_a - s_i$ part of $T_a$ is its ``negative'' (or red) cell and the rightmost $q_a - s_i$
part is its ``positive'' (or blue) cell. The part left in between the positive and negative parts,
which is of length $2s_i$, is its ``zero'' (or green) cell. All 3 cells contain the same vertex $b^*$ where $\Comp_i = \{a,b^*\}$.

\medskip

\item for $b \in \Comp_i$, compute $2s_i = 1-p_i$. Since $\pm \alpha_b \le p_i \le 1$, we have $s_i \ge 0$; also
$s_i \le q_b$ and $s_i \le 1 - q_b$ where $q_b = (1 - \alpha_b)/2$.
  The array $T_b$ gets divided into three cells as shown in Fig.~\ref{new-Cj-2:example}.

\begin{figure}[h]
\centerline{\resizebox{0.34\textwidth}{!}{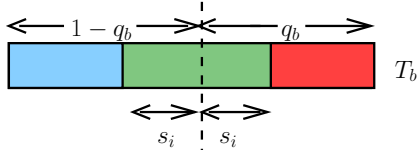}}
\caption{The blue part (of length $1-q_b - s_i$) is the positive cell of $T_b$ and the red part (of length $q_b - s_i$) is the negative cell of $T_b$. The middle or green part (of length $2s_i$) is its zero cell.}
\label{new-Cj-2:example}
\end{figure}
  
The leftmost $1 - q_b - s_i$ part of
  $T_b$ is its ``positive'' (or blue) cell and the rightmost $q_b - s_i$ part is its ``negative'' (or red) cell. 
The part left in between these two parts, which is of length $2s_i$, is its ``zero'' (or green) cell.
All 3 cells contain the same vertex $a^*$ where $\Comp_i = \{a^*,b\}$.
\end{itemize}

In order to find the popular matchings that $\vec{x}$ is a convex combination of, we build the table $T$ and define $M_t$ exactly 
as done in the proof of Theorem~\ref{thm:polytope}. It follows from the same arguments as before that $M_t$ is a matching.
In particular, for $k+1 \le i \le \ell$, we have $|\Comp_i| = 2$, say $\Comp_i = \{u,v\}$, and so $v$ is the only entry in the entire
array $T_u$ and similarly, $u$ is the only entry in the entire array $T_v$. 

We now need to show that $M_t$ is popular.
For this, we define a vector $\vec{\alpha}^t \in \{0,\pm 1\}^n$: for $i \in \{1,\ldots,k\} \cup \{\ell+1,\ldots,h\}$,
the assignment of $\alpha^t_u$-values is exactly the same as in the proof of Theorem~\ref{thm:polytope}.
For $i \in \{k+1,\ldots,\ell\}$ and each $u \in \Comp_i$ do:
\begin{itemize}
\item set $\alpha^t_u = 1$ if the cell $T_u[t]$ is positive (or blue) 
\item set $\alpha^t_u = 0$ if the cell $T_u[t]$ is zero (or green) 
\item set $\alpha^t_u = -1$ if the cell $T_u[t]$ is negative (or red)
\end{itemize}

We will now show that $\sum_{u \in A \cup B}\alpha^t_u = 0$. The only new step is to show that
$\alpha^t_a + \alpha^t_b = 0$ for $(a,b) \in M_t$, where $\Comp_i = \{a,b\}$, i.e., $i \in \{k+1,\ldots,\ell\}$.
Here we have $\alpha_a + \alpha_b = \wt_x(a,b)= 0$ and this is because $x_{ab} = 1$.

So $q_a = (1 + \alpha_a)/2 = (1 - \alpha_b)/2 = q_b$ and this implies the length of the {\em positive} (or blue) 
cell in $T_a$, which is  $q_a - s_i$ (see Fig.~\ref{new-Cj-1:example}), equals the length of the {\em negative} (or red) cell 
in  $T_b$, which is $q_b - s_i$, (see Fig.~\ref{new-Cj-2:example}).
   \begin{itemize}
   \item Hence the positive cell of $T_a$ is perfectly
     aligned with the negative cell of $T_b$. So if $L(t)$ goes through the positive cell of $T_a$, 
   i.e. if $\alpha^t_a = 1$, then $\alpha^t_b = -1$.

    \item Similarly, the negative cell of $T_a$, which is of length 
    $1 - q_a - s_i$, is perfectly aligned with the positive cell of $T_b$, 
    which is of length $1 - q_b - s_i$. So if $L(t)$ goes 
    through the negative cell of $T_a$, i.e. if $\alpha^t_a = -1$, then $\alpha^t_b = 1$.
   
    \item Thus the zero cells in $T_a$ and $T_b$ are perfectly aligned with 
     each other. So if $L(t)$ goes through the zero cell of $T_a$, i.e. if $\alpha^t_a = 0$, 
     then $\alpha^t_b = 0$.
 \end{itemize}
Hence $\alpha^t_a + \alpha^t_b = 0$ for all $(a,b) \in M_t$ and so $\sum_{u \in A \cup B}\alpha^t_u = 0$.

\smallskip

It is easy to see that $\alpha^t_u \ge \wt_{M_t}(u,u)$ for all $u \in A\cup B$. What is left to show is that
$\alpha^t_a + \alpha^t_b \ge \wt_{M_t}(a,b)$ for all $(a,b) \in E$. Let $a \in \Comp_i$ and $b \in \Comp_j$.
When both $i$ and $j$ are in $\{1,\ldots,k\} \cup \{\ell+1,\ldots,h\}$, the proof of Lemma~\ref{lemma:mat3} shows that 
the edge covering constraint holds. Lemmas~\ref{lemma:mat4} and \ref{lemma:mat5} below 
show that the edge covering constraint holds when one or both the indices are in $\{k+1,\ldots,\ell\}$.

This completes the proof that $M_t$ is a popular matching in $G$,
in particular, $M_t \in {\cal P'}(\vec{r})$. The rest of the argument that
$\vec{x}$ is a convex combination of matchings that belong to $\mathcal{P'}(\vec{r})$ is exactly the same as given in the
proof of Theorem~\ref{thm:polytope}. Thus we can conclude that $\mathcal{S}(\vec{r}) = \mathcal{P}'(\vec{r})$.  \qed
\end{proof}

\begin{lemma}
  \label{lemma:mat4}
  Let $(a,b)\in E$ with $a\in\Comp_i$ and $b\in\Comp_j$. Suppose one of $i,j$ is in $\{1,\ldots,k\} \cup \{\ell+1,\ldots,h\}$ 
   and the other is in $\{k+1,\ldots,\ell\}$. Then $\alpha^t_a + \alpha^t_b \ge \wt_{M_t}(a,b)$.
\end{lemma}
\begin{proof}
  Assume without loss of generality that $i \in \{k+1,\ldots,\ell\}$. So $\Comp_i = \{a,b^*\}$.
  Since $a$ and $b$ are in different components in $F_G$, the edge $(a,b) \notin E_F$ and so
  $x_{ab} = 0$. 
  The index $j \in \{1,\ldots,k\}\cup\{\ell+1,\ldots,h\}$, so $r_j$ is 0 or 1.

\smallskip
  
\noindent{\em Case 1.} Suppose $r_j = 0$. So $\alpha_b = 0$. Suppose $b \in U$, i.e., $b$ is unmatched in $\vec{x}$. Then the constraint 
$\alpha_a + \alpha_b  \ge \wt_x(a,b) + |r_i-r_j|$ for $(a,b)$ becomes $\alpha_a \ge p_i$ since $\wt_x(a,b) = 0$. 
We also have 
$\alpha_a \le p_i$ (the revised constraint~(\ref{new-constr3})). Hence $\alpha_a = p_i$ and this means $2s_i = 1 - \alpha_a$, i.e.,
$s_i = (1-\alpha_a)/2 = 1-q_a$ (see Fig.~\ref{new-Cj-1:example}). Then there is {\em no} negative (or red) cell in the entire array $T_a$.
In other words, $\alpha^t_a \ge 0$ throughout the array $T_a$ and so $\alpha^t_a + \alpha^t_b \ge 0 = \wt_{M_t}(a,b)$.

Hence let us assume that $b \in B'$. So $\wt_x(a,b) = 2(\sum_{b' \prec_a b}x_{ab'} + \sum_{a' \prec_b a}x_{a'b} -1)$.
Thus the constraint $\alpha_a + \alpha_b - |r_i-r_j| \ge \wt_x(a,b)$ becomes: (where $\alpha_a = 2q_a-1$)
\begin{eqnarray*}
(2q_a - 1) - p_i & \ \ge \ & 2\left(\sum_{b' \prec_a b}x_{ab'} + \sum_{a' \prec_b a}x_{a'b} -1\right)\\ 
q_a + (1-p_i)/2 & \ \ge \ & \sum_{b' \prec_a b}x_{ab'} + \sum_{a' \prec_b a}x_{a'b}\\
q_a + s_i  & \ \ge \ & \sum_{b' \prec_a b}x_{ab'} + \sum_{a' \prec_b a}x_{a'b} \ \ \ \ \ \ \ (\text{since}\ s_i = (1-p_i)/2)
\end{eqnarray*}

Note that $q_a + s_i$ is the sum of lengths of the positive (or blue) and zero (or green) cells of 
$T_a$ (see Fig.~\ref{new-Cj-3:example}) and so $q_a + s_i \le 1$. In these two 
cells of $T_a$, $\alpha^t_a$ is either 1 or 0.

\begin{figure}[h]
\centerline{\resizebox{0.34\textwidth}{!}{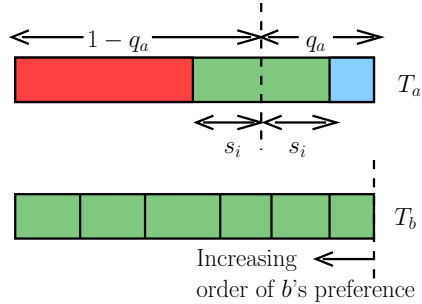}}
\caption{The vertex $a$ is matched to the same partner in the entire array $T_a$ and $b$'s increasing order of
partners in $\vec{x}$ starts from the right end of its array $T_b$.}
\label{new-Cj-3:example}
\end{figure}

Suppose $a$ prefers $b^*$ to $b$. Then $\sum_{b' \prec_a b}x_{ab'} = 0$, so 
$q_a + s_i \ge \sum_{a' \prec_b a}x_{a'b}$. 
So  while $b$ is matched to a neighbor worse than $a$, $\alpha^t_a \ge 0$.
Note that  $\alpha^t_b = 0$ throughout the array $T_b$.  
Thus $\alpha^t_a + \alpha^t_b \ge 0 \ge \wt_{M_t}(a,b)$.

Suppose $a$ prefers $b$ to $b^*$. Then $\sum_{b' \prec_a b}x_{ab'} = 1$ and so $q_a + s_i \ge 1$. 
This means $q_a + s_i = 1$ and so $\sum_{a' \prec_b a}x_{a'b} = 0$. That is, $b$ prefers each of its 
partners in the array $T_b$ to $a$ and so $\wt_{M_t}(a,b) = 0$. Also $q_a + s_i = 1$ implies that there is {\em no} negative (or red) 
cell in $T_a$. Thus $\alpha^t_a \ge 0$ throughout the array $T_a$ and so $\alpha^t_a + \alpha^t_b \ge 0 = \wt_{M_t}(a,b)$.

\smallskip

\noindent{\em Case 2.} Suppose $r_j = 1$. So $\alpha_b = 1 -2q_b$ and the constraint $\alpha_a + \alpha_b - |r_i-r_j| \ge \wt_x(a,b)$ 
becomes:
\begin{eqnarray*}
(2q_a - 1) + (1 - 2q_b) - (1-p_i) & \ \ \ge \ \ & 2\left(\sum_{b' \prec_a b}x_{ab'} + \sum_{a' \prec_b a}x_{a'b} -1\right)\\
q_a - q_b - (1-p_i)/2 & \ \ \ge \ \ & \sum_{b' \prec_a b}x_{ab'} + \sum_{a' \prec_b a}x_{a'b} -1\\
(q_a - s_i) + (1 - q_b) & \ \ \ge \ \ & \sum_{b' \prec_a b}x_{ab'} + \sum_{a' \prec_b a}x_{a'b}.
\end{eqnarray*}

Note that $q_a - s_i$ is the length of the positive (or blue) cell in $T_a$ 
(see Fig.~\ref{new-Cj-4:example}).  
Similarly, $1 - q_b$ is the length of the blue sub-array of $T_b$. 

Suppose $a$ prefers $b^*$ to $b$. Then  $\sum_{b' \prec_a b}x_{ab'} = 0$ and so
$\sum_{a' \prec_b a}x_{a'b} \le (q_a - s_i) + (1 - q_b)$ and this is the sum of lengths of blue 
sub-arrays of $T_a$ and $T_b$. Consider traversing the array $T_b$ starting from the dashed line 
separating its blue sub-array from its red sub-array in a right-to-left orientation that wraps around.
The sum of length of the cells where $b$ is matched to neighbors worse than $a$ is 
$\sum_{a' \prec_b a}x_{a'b}$. This is at most the sum of lengths of blue sub-arrays in $T_a$ and $T_b$. 
Thus while $b$ is matched to a neighbor worse than $a$, we have $\alpha^t_a + \alpha^t_b \ge 0 = \wt_{M_t}(a,b)$
and when $b$ is matched to a neighbor better than $a$, we have $\alpha^t_a + \alpha^t_b \ge -2 = \wt_{M_t}(a,b)$
since $\alpha^t_a \ge -1$ and  $\alpha^t_b \ge -1$.

\begin{figure}[h]
\centerline{\resizebox{0.41\textwidth}{!}{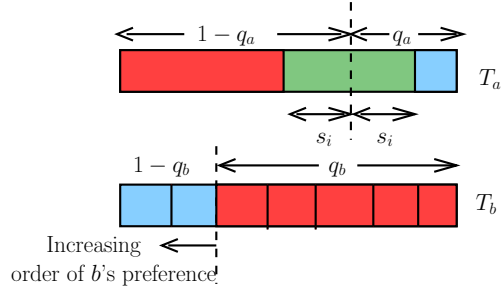}}
\caption{The vertex $a$ is matched to the same partner in the entire array $T_a$ and $b$'s 
increasing order of partners in $T_b$ starts from the dashed line in right to left orientation.} 
\label{new-Cj-4:example}
\end{figure}

Suppose $a$ prefers $b$ to $b^*$. Then  $\sum_{b' \prec_a b}x_{ab'} = 1$ and so
$\sum_{a' \prec_b a}x_{a'b} \le (q_a - s_i) + (1 - q_b) -1$. Since $\sum_{a' \prec_b a}x_{a'b} \ge 0$, 
this means the sum of lengths of  blue sub-arrays in $T_a$ and $T_b$ is at least 1 and 
$\sum_{a' \prec_b a}x_{a'b}$ is bounded by how much $(q_a - s_i) + (1 - q_b)$ exceeds 1. Since the blue
sub-array in $T_b$ begins from its left end and the blue cell in  $T_a$ starts from its right end,
for any $t \in [0,1)$,  at least one of the cells $T_a[t], T_b[t]$
 is blue. Moreover, when $b$  prefers $a$ to its neighbors in $T_b$,
 then {\em both} $T_a[t]$ and $T_b[t]$ are blue. So while $b$ is matched to a neighbor better than $a$,
 we have $\alpha^t_a + \alpha^t_b \ge 0 = \wt_{M_t}(a,b)$ and 
 when $b$ is matched to a neighbor worse than $a$, we have $\alpha^t_a + \alpha^t_b = 2 = \wt_{M_t}(a,b)$. \qed
\end{proof}

\begin{lemma}
  \label{lemma:mat5}
  Let $(a,b)$ be an edge in $\Comp_i \times\Comp_j$, where $i,j \in \{k+1,\ldots,\ell\}$. Then
  $\alpha^t_a + \alpha^t_b \ge \wt_{M_t}(a,b)$.
\end{lemma}
\begin{proof}
Here both $\Comp_i$ and $\Comp_j$ have size 2. When $i = j$, we showed that $\alpha^t_a + \alpha^t_b = 0 = \wt_{M_t}(a,b)$ in the proof 
of $\sum_{u\in A \cup B}\alpha^t_u = 0$.

So we now assume $i \ne j$ and so $x_{ab} = 0$. 
The constraint $\alpha_a + \alpha_b - |r_i-r_j| \ge \wt_x(a,b)$ becomes the following two constraints, where 
$\alpha_a = 2q_a - 1$ and $\alpha_b = 1 - 2q_b$.

$(2q_a - 1) + (1 - 2q_b) - (p_i-p_j) \ \ge \ 2\,(\sum_{b' \prec_a b}x_{ab'} + \sum_{a' \prec_b a}x_{a'b} -1)$ and

$(2q_a - 1) + (1 - 2q_b) + (p_i-p_j) \ \ge \ 2\,(\sum_{b' \prec_a b}x_{ab'} + \sum_{a' \prec_b a}x_{a'b} -1)$.

\smallskip

\noindent Let $s_i = (1-p_i)/2$ and $s_j = (1-p_j)/2$. The above two constraints imply the following two constraints respectively: 
\begin{eqnarray}
(q_a + s_i) + (1 - q_b - s_j) & \ge & \sum_{b' \prec_a b}x_{ab'} + \sum_{a' \prec_b a}x_{a'b}\label{newer-constr1}\\
(q_a - s_i) + (1 - q_b + s_j) & \ge & \sum_{b' \prec_a b}x_{ab'} + \sum_{a' \prec_b a}x_{a'b}\label{newer-constr2}
\end{eqnarray}

Note that the length of the blue cell in $T_a$ is $q_a-s_i$ while $q_a+s_i$ is the sum of the 
lengths of the blue and green cells in $T_a$ (see the top 2 arrays in Fig.~\ref{new-Cj-5:example}). Similarly, the 
length of the blue cell in $T_b$ is $1 - q_b - s_j$ while $1 - q_b + s_j$ is the sum of the 
lengths of the blue and green cells in $T_b$  (see the bottom 2 arrays in  Fig.~\ref{new-Cj-5:example}). 
We will consider 3 cases here.

\smallskip

\noindent{\em Case~1.} Both $a$ and $b$ prefer their partners in $\vec{x}$ to each other. This is 
the easiest case. Here $\wt_{M_t}(a,b) = -2$ and since $\alpha^t_u \ge -1$ for all vertices $u$,
we have $\alpha^t_a + \alpha^t_b \ge -2 = \wt_{M_t}(a,b)$.

\begin{figure}[h]
\centerline{\resizebox{0.66\textwidth}{!}{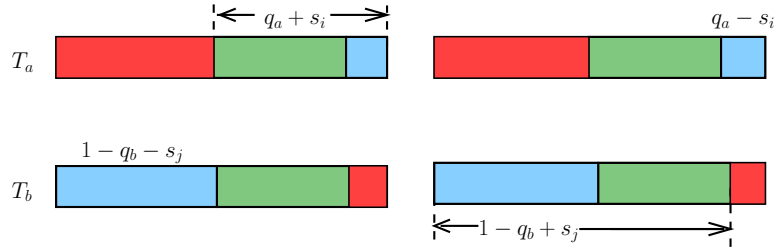}}
\caption{The length of the blue cell in $T_a$ (similarly, $T_b$) is $q_a-s_i$ 
(resp., $1-q_b-s_j$). The length of the blue $+$ green cells in $T_a$ (similarly, $T_b$) is 
$q_a+s_i$ (resp., $1-q_b+s_j$). }
\label{new-Cj-5:example}
\end{figure}

\noindent{\em Case~2.} Exactly one of $a, b$ prefers its partner in $\vec{x}$ to the other. 
So $\wt_{M_t}(a,b) = 0$ and the right side of constraints~(\ref{newer-constr1}) and (\ref{newer-constr2}) is 1.

Constraint~(\ref{newer-constr1}) means that the length of the 
(blue $+$ green) cells of $T_a$ added to the length of the blue cell of $T_b$ 
is at least 1. Hence the length of the blue cell of $T_b$ is at least the length of the 
red cell in $T_a$ (see Fig.~\ref{new-Cj-5:example}).

Constraint~(\ref{newer-constr2}) means that the length of the (blue $+$ green) cells of $T_b$ 
added to the length of the blue cell of $T_a$ is at least 1. 
Hence the length of the blue cell of $T_a$ is at least the length of the red cell in 
$T_b$ (see Fig.~\ref{new-Cj-5:example}).

Thus for any $t \in [0,1)$, it is the case that either (i)~at least one of $T_a[t], T_b[t]$
is blue or (ii)~{\em both} the cells $T_a[t]$ and $T_b[t]$
are green. Thus we have $\alpha^t_a + \alpha^t_b \ge 0 = \wt_{M_t}(a,b)$.

\smallskip

\noindent{\em Case 3.} Both $a$ and $b$ prefer each other to their partners in $\vec{x}$. So
$\wt_{M_t}(a,b) = 2$ and the right side of constraints~(\ref{newer-constr1}) and (\ref{newer-constr2}) 
is also 2.

Since each of $(q_a + s_i), (q_a - s_i), (1 - q_b + s_j), (1 - q_b - s_j)$ is at most 1, 
it follows that $q_a - s_i = 1$ and $1 - q_b - s_j = 1$. 
Thus the entire array $T_a$ is blue and similarly, the entire array $T_b$ is also blue.
Hence $\alpha^t_a = \alpha^t_b = 1$. Thus we have $\alpha^t_a + \alpha^t_b = 2 = \wt_{M_t}(a,b)$. \qed 
\end{proof}

\medskip

\noindent{\em Acknowledgment.} Thanks to Chien-Chung Huang for useful discussions on strongly dominant matchings.

\section*{Appendix}
\noindent{\bf Mixed matchings.}
Consider the roommates instance $G$ that is a triangle on the 3 vertices $a$, $b$, $c$ with 
the following preference lists: $a$ prefers $b$ to $c$ while $b$ prefers $c$ to $a$ and $c$ prefers $a$ to $b$. 
This instance has no {\em mixed matching} that
is stable. A mixed matching $\Pi$ is equivalent to a point $\vec{x}$ in the
matching polytope of $G$. A mixed matching $\Pi$ is stable if the corresponding point $\vec{x}$ satisfies the following stability 
constraint for every edge $(u,v)$:
\[ x_{uv} + \sum_{v':v' \succ_u v}x_{uv'} + \sum_{u':u' \succ_v u}x_{u'v} \ \ge \ 1.\]
It can be checked that there is no point $\vec{x}$ in the matching polytope of $G$ that satisfies the stability constraints for all 
edges. The mixed matching 
$\Pi = \{(M_1,1/3),(M_2,1/3),(M_3,1/3)\}$, where $M_1 = \{(a,b)\}$, $M_2 = \{(b,c)\}$, $M_3 = \{(c,a)\}$, is popular in this instance.

\subsection*{Irving's algorithm}
Given a roommates instance $G = (V,E)$ with strict preferences, Irving's algorithm~\cite{Irv85} determines if $G$ admits a stable matching 
or not and if so, returns one. Irving's algorithm assumed $G$ to be a complete graph, however the algorithm easily generalizes to
non-complete graphs as well and hence we will not assume $G$ to be complete.

Irving's algorithm consists of 2 phases:
\begin{enumerate}
\item   In the first phase, we consider the bipartite graph
  $G^* = (V \cup V', E')$ where $V' = \{u': u \in V\}$. So $G^*$ has 2 copies $u$ and $u'$ of each vertex $u \in V$, one on either side
  of the graph.   The edge set $E' = \{(u,v'): (u,v) \in E\}$. 
  
 Run Gale-Shapley algorithm on $G^*$ with vertices in $V$ proposing and those in $V'$ disposing.
 Let $M^*$ be the resulting matching. If $(u,v') \in M^*$ then prune the edge set $E$ of $G$ as follows:
 \begin{itemize}
     \item delete all neighbors ranked worse than $u$ from the preference list of $v$
     \item make the adjacency lists symmetric so that if $v$ deletes $u$ from its list then $u$ also deletes $v$ from its list.
 \end{itemize}
\item If the {\em reduced} adjacency list of every vertex which received at least 1 proposal consists of a single neighbor then
  the resulting edge set $E$ defines a stable
  matching $M$. Else the adjacency lists are further reduced by eliminating ``rotations''.
  \begin{itemize}
  \item A rotation $R = \{(a_0,b_0),(a_1,b_1),\ldots,(a_{k-1},b_{k-1})\}$ is a set of edges such that for all $i \in \{0,\ldots,k-1\}$, 
  the vertex  $b_i$ is $a_i$'s most
  preferred neighbor in its reduced preference list (thus $a_i$ would be $b_i$'s least preferred neighbor in its reduced preference list);
  moreover, the second person on $a_i$'s reduced preference list is $b_{i+1}$ (here $b_k = b_0$).

      \item The second phase of Irving's algorithm identifies such rotations and deletes them. The crucial observation here is that if $G$
  admits a stable matching then so does $G\setminus R$.
    
      \item This step of eliminating rotations continues till either the updated reduced adjacency list
        of every vertex consists of a single neighbor or the updated reduced adjacency list of some vertex that received at least 1 proposal
        in the first phase is empty. In the former case, the resulting edge set is a
  stable matching and in the latter case, $G$ has no stable matching.
  \end{itemize} 
\end{enumerate}  

Consider Irving's algorithm in the roommates instance $G'$ that corresponds to the instance $G$ on 4 vertices $a, b, c, d$ described in 
Section~\ref{intro}. The reduced  adjacency lists at the end of the first phase are as
follows:
\[ a: \ c^- \succ d^- \ \ \ \ \ \ \ \ b: \ d^- \succ c^+ \ \ \ \ \ \ \ \ c: \ b^- \succ a^+ \ \ \ \ \ \ \ \ d: \ a^+ \succ b^+\]
Eliminating the rotation $\{(a^+,c^-),(b^+,d^-)\}$ yields the matching $M'_1 = \{(a^+,d^-),(b^-,c^+)\}$.
Instead, we could have 
eliminated the rotation $\{(c^+,b^-),(d^-,a^+)\}$. This yields the matching $M'_2 = \{(a^+,c^-),(b^+,d^-)\}$.


\begin{thebibliography}{22}
\bibitem{AIKM07}
D.J. Abraham, R.W. Irving, T.~Kavitha, and K.~Mehlhorn.
\newblock {\em Popular matchings.}
\newblock SIAM Journal on Computing, 37(4):~1030--1045, 2007.

\bibitem{BIM10}
P. Biro, R. W. Irving, and D. F. Manlove.
\newblock {\em Popular Matchings in the Marriage and Roommates Problems.}
\newblock In the 7th International Conference in Algorithms and Complexity (CIAC):~97--108, 2010.

\bibitem{BK17}
  F. Brandl and T. Kavitha.
  \newblock {\em Popular matchings with multiple partners.}
  \newblock In the 37th Foundations of Software Technology and Theoretical Computer Science (FSTTCS), 2017.

\bibitem{wiki-condorcet}
Condorcet method.
\newblock \url{https://en.wikipedia.org/wiki/Condorcet_method}

\bibitem{SocialChoice17}
\'{A}. Cseh.
\newblock {\em Popular Matchings.}
\newblock In Trends in Computational Social Choice, Edited by Ulle Endriss,
COST (European Cooperation in Science and Technology):~105--122, 2017.
  
\bibitem{CHK15}
\'{A}. Cseh, C.-C. Huang, and T. Kavitha.
\newblock {\em Popular matchings with two-sided preferences and one-sided ties}. 
\newblock In the 42nd International Colloquium on Automata, Languages, and Programming (ICALP): Part~I,~367--379, 2015.


\bibitem{CK16}
\'{A}. Cseh and T. Kavitha.
\newblock {\em Popular edges and dominant matchings.}
\newblock In the 18th International Conference on Integer Programming and Combinatorial Optimization (IPCO):~138--151, 2016. 



\bibitem{Feder92}
T. Feder.
\newblock {\em A new fixed point approach for stable networks and stable marriages.}
\newblock Journal of Computer and System Sciences, 45(2):~233-284, 1992. 

\bibitem{Fed94}
T.~Feder.
\newblock {\em Network flow and 2-satisfiability.}
\newblock Algorithmica, 11(3):291--319, 1994.

\bibitem{Fle03}
T. Fleiner.
\newblock {\em A fixed-point approach to stable matchings and some applications. }
\newblock Mathematics of Operations Research, 28(1):~103-126, 2003.

\bibitem{GS62}
D.~Gale and L.S. Shapley.
\newblock {\em College admissions and the stability of marriage.}
\newblock American Mathematical Monthly, 69(1):~9--15, 1962.

\bibitem{GS85}
D.~Gale and M.~Sotomayor.
\newblock {\em Some remarks on the stable matching problem.}
\newblock Discrete Applied Mathematics, 11(3):~223--232, 1985.

\bibitem{Gar75}
P. G\"ardenfors.
\newblock {\em Match making: assignments based on bilateral preferences.}
\newblock Behavioural Sciences, 20(3):~166--173, 1975.

\bibitem{GI89}
D. Gusfield and R. W. Irving.
\newblock {\em The Stable Marriage Problem: Structure and Algorithms.}
\newblock MIT Press, Boston, MA 1989.

\bibitem{Hirakawa-MatchUp15}
M. Hirakawa, Y. Yamauchi, S. Kijima, and M. Yamashita.
\newblock {\em On The Structure of Popular Matchings in The Stable Marriage Problem - Who Can Join a Popular Matching?}
\newblock In the 3rd International Workshop on Matching Under Preferences (MATCH-UP), 2015.

\bibitem{HK11}
C.-C.~Huang and T.~Kavitha.
\newblock {\em Popular matchings in the stable marriage problem.} 
\newblock Information and Computation, 222:~180--194, 2013.
(Special issue on ICALP 2011.)

\bibitem{HK13}
C.-C.~Huang and T.~Kavitha.
\newblock{\em Near-Popular Matchings in the Roommates Problem.}
\newblock SIAM Journal on Discrete Mathematics, 27(1): 43--62, 2013.

\bibitem{HK17}
C.-C.~Huang and T.~Kavitha.
\newblock {\em Popularity, Self-Duality, and Mixed matchings.} 
\newblock In the 28th ACM-SIAM Symposium on Discrete Algorithms (SODA): 2294-2310, 2017.

\bibitem{Irv85}
R. W. Irving.
\newblock {\em An efficient algorithm for the stable roommates problem.}
\newblock Journal of Algorithms, 6:~577--595, 1985. 

\bibitem{ILG87}
R. W.~Irving, P.~Leather, and D.~Gusfield.
\newblock {\em An efficient algorithm for the ``optimal'' stable marriage.}
\newblock Journal of the ACM, 34(3):~532--543, 1987.

\bibitem{KMN09}
T. Kavitha, J. Mestre, and M. Nasre.
\newblock {\em Popular mixed matchings.}
\newblock Theoretical Computer Science, 412(24):~2679--2690, 2011.

\bibitem{Kav12}
T.~Kavitha.
\newblock {\em A size-popularity tradeoff in the stable marriage problem.}
\newblock SIAM Journal on Computing, 43(1):~52--71, 2014.

\bibitem{Kav16}
T.~Kavitha.
\newblock {\em Popular half-integral matchings.}
\newblock In the 43rd International Colloquium on Automata, Languages, and Programming (ICALP):~22.1-22.13, 2016.

\bibitem{Roth84}
  A. E. Roth.
\newblock {\em The evolution of the labor market for medical interns and residents. A case study in game theory.}
\newblock Journal of Political Economy, 92:~991-1016, 1984.

\bibitem{Rot92}
U.~G.~Rothblum.
\newblock {\em Characterization of stable matchings as extreme points of a polytope.}
\newblock Mathematical Programming, 54:~57--67, 1992.

\bibitem{Sub94}
  A. Subramanian.
\newblock {\em A new approach to stable matching problems.}
\newblock SIAM Journal on Computing, 23(4):~671--700, 1994.


\bibitem{TS98}
C.-P.~Teo and J.~Sethuraman.
\newblock {\em The geometry of fractional stable matchings and its applications.}
\newblock Mathematics of Operations Research, 23(4):~874--891, 1998.

\bibitem{VV89}
J.~H. Vande Vate.
\newblock {\em Linear programming brings marital bliss.}
\newblock Operations Research Letters, 8(3):~147--153, 1989.
\end{thebibliography}
\end{document}